\PassOptionsToPackage{table}{xcolor}
\documentclass[a4paper, 11pt]{article}

\title{Descriptive complexity for distributed computing with circuits} 

\usepackage[textwidth=151mm, hcentering, vmargin={24mm,32.5mm}, foot=20mm]{geometry}

\usepackage{amsmath,amssymb,amsfonts}
\usepackage{amsthm}
\usepackage{hyperref}
\usepackage{verbatim}

\usepackage[utf8]{inputenc}
\usepackage[T1]{fontenc}
\usepackage{amsmath, amssymb, amsthm}
\usepackage{enumitem}
\usepackage{tikz}
\usepackage{lipsum}
\usepackage{pdfpages}
\usepackage{multicol} 
\usepackage{stmaryrd}

\makeatletter
\def\Ddots{\mathinner{\mkern1mu\raise\p@
\vbox{\kern7\p@\hbox{.}}\mkern2mu
\raise4\p@\hbox{.}\mkern2mu\raise7\p@\hbox{.}\mkern1mu}}
\makeatother

\usepackage{arydshln}
\usepackage{makecell}

\usepackage{mathtools}

\newcommand{\N}{\mathbb N}
\newcommand{\Z}{\mathbb Z}

\newcommand{\cA}{\mathcal{A}}

\newcommand{\cK}{\mathcal{K}}

\newcommand{\cP}{\mathcal{P}}

\newcommand{\cT}{\mathcal{T}}

\newcommand*{\abs}[1]{\lvert#1\rvert}   

\newcommand{\head}{\mathrm{HEAD}}

\newcommand{\msc}{\mathrm{MSC}}
\newcommand{\cmsc}{\mathrm{CMSC}}
\newcommand{\mmsc}{\mathrm{MMSC}}
\newcommand{\mpmsc}{\mathrm{MPMSC}}
\newcommand{\ml}{\mathrm{ML}}
\newcommand{\mdt}{\mathrm{mdt}}
\newcommand{\mdi}{\mathrm{mdi}}
\newcommand{\md}{\mathrm{md}}

\newcommand{\subs}{\mathrm{SUBS}}

\newcommand{\ordo}{\mathcal{O}}

\newcommand{\ID}{\mathrm{ID}}
\newcommand{\fmpa}{\mathrm{FMPA}}
\newcommand{\mpc}{\mathrm{MPC}}

\theoremstyle{plain}
\newtheorem{theorem}{Theorem}[section]
\newtheorem{lemma}[theorem]{Lemma}
\newtheorem{corollary}[theorem]{Corollary}

\theoremstyle{definition}

\newtheorem{definition}[theorem]{Definition}

\usepackage[utf8]{inputenc} 
\usepackage[T1]{fontenc}    
\usepackage{hyperref}       
\usepackage{url}            
\usepackage{booktabs}       
\usepackage{amsfonts}       
\usepackage{microtype}      
\usepackage{graphicx}

\date{} 					

\usepackage{authblk}

\author[1]{Veeti Ahvonen}
\author[1]{Damian Heiman}
\author[1]{Lauri Hella} 
\author[1,2]{Antti Kuusisto}

\affil{Tampere University, Finland}
\affil[2]{University of Helsinki, Finland}

\begin{document}
\maketitle
\begin{center}
    { email: \url{firstname.lastname@tuni.fi}} 
    \vspace{1em}
\end{center}

\begin{abstract}
\noindent
We consider distributed algorithms in the realistic scenario where distributed message passing is operated via circuits. We show that within this setting, modal substitution calculus MSC captures the expressive power of circuits. The translations between circuits and MSC-programs are linear in both directions. Furthermore, we show that the colouring algorithm based on Cole-Vishkin can be specified via logarithmic size programs. 
\end{abstract}

\textbf{Keywords:} distributed computing, descriptive complexity, logic, graph colouring

\section{Introduction}

Distributed computing concerns computation in networks and relates directly to various fields of study including, inter alia, cellular automata and neural networks. In this paper, we study distributed systems based on circuits. A distributed system is a labeled directed graph (with self-loops allowed) where nodes communicate by sending messages to each other. In each communication round, a node sends a message to its neighbours and updates its state based on (1) its own previous state and (2) the messages received from the neighbours.

Descriptive complexity of distributed computing was initiated in \cite{weak_models}, which characterized classes of constant time distributed algorithms via modal logics. The constant-time assumption was lifted in \cite{Kuusisto13}, which showed that the expressive power of \emph{finite message passing automata} (FMPAs) is captured by \emph{modal substitution calculus} MSC, which is an extension of modal logic by Datalog-style rules. The papers \cite{weak_models} and \cite{Kuusisto13} did not consider \emph{identifiers}, i.e., $\mathrm{ID}$-numbers analogous to IP-addresses.

In this paper, we study distributed computing based on circuits in a scenario with identifiers. Each node runs a copy of the same circuit $C$. In each communication round, the node sends its current bit string $s$ to its neighbours and updates to a new string $s'$ by feeding $s$ and the
strings $s_1,\dots , s_m$ sent by the neighbours to $C$ (letting $s'$ be the output of $C$). This is a realistic model of distributed computing, which also takes \emph{local computation}---the computation steps of the circuit---into account. Typically in distributed computing, only communication steps count. Since we study distributed systems, we call our circuits \emph{message passing circuits}, or $\mpc$s, although formally they are just plain circuits.

We establish an exact match between this circuit-based model and the logic $\mathrm{MSC}$. Unlike earlier works on descriptive complexity of
distributed computing, we work in the circuit-style
paradigm where an algorithm is specified via an \emph{allocation function} $F$ that produces, in the simplest case, for each input $n\in \mathbb{Z}_+$, a circuit $F(n)$ that operates on all distributed systems (i.e., labeled directed graphs, or Kripke models) of size $n$. As one of our main results, we prove that programs of the $\msc$-logic and constant fan-in message passing circuits translate to each other with only a \emph{linear} blow-up in size. 
Thus, we can work interchangeably with circuit allocation functions and $\msc$-program allocation functions. The related formal statements are as follows, with $\Pi$ denoting the set of proposition symbols considered (including ones for $\mathrm{ID}$-bits) while $\Delta$ is a degree bound for graphs.

\smallskip

\noindent
\textbf{Theorem\ \ref{thr:mpc_to_msc}.}\ \emph{Given an $\mpc$ of size $m$ for $(\Pi, \Delta)$, we can construct an equivalent $\Pi$-program of $\msc$. For a constant bound $c$ for the fan-in of $\mpc$s, the size of the program is $\ordo(m)$.}

\smallskip

\noindent
\textbf{Theorem\ \ref{thr:msc_to_mpc}.}\ \emph{Given $\Pi$, $\Delta$ and a $\Pi$-program of $\msc$ of size $m$, we can construct an equivalent $\mpc$ for $(\Pi, \Delta)$ of size $\ordo(\Delta m + \abs{\Pi})$.}

\smallskip

We are especially interested in the feasible scenario where $F(n)$ is a circuit of size $\mathcal{O}(\log n)$. 
From the above results, we can prove that, for a constant $\Delta$ and constant fan-in bound, if we have an allocation function producing log-size circuits, we also have an allocation function for log-size programs, and vice versa. 
We put this into use by demonstrating that for
graphs of degree bound $\Delta$, we can produce programs of size $\mathcal{O}(\log n)$ that compute a $(\Delta + 1)$-colouring via a \emph{Cole-Vishkin} \cite{goldberg1987parallel} style approach---implying also an analogous result for circuits. 

Generally, the circuit-based approach suits well for studying the \emph{interplay of local computation and message passing}. While important, such effects have received relatively little attention in studies on distributed computing. We provide a range of related results.

\medskip

\noindent
\textbf{Related work.} As already mentioned, descriptive complexity of distributed computing has been largely initiated in \cite{weak_models}, which characterizes a range of related complexity classes via modal logics. It is shown, for example, that graded modal logic captures the class $\mathrm{MB}(1)$ containing problems solvable in constant time by algorithms whose recognition capacity is sufficient all the way up to distinguishing between multisets of incoming messages, but no further. In the paper, the link to logic helps also in \emph{separating} some of the studied classes. The constant-time limitation is lifted in \cite{Kuusisto13}, which shows that finite distributed message passing automata (FMPAs) correspond to modal substitution calculus MSC, which is the logic studied also in the current paper. The work on MSC is extended in \cite{reitericalp}, which proves that while MSC corresponds to synchronized automata, the $\mu$-fragment of the modal $\mu$-calculus similarly captures asynchronous distributed automata. 

Distributed computing with identifiers has been studied from the point of view of logic earlier in \cite{identifiers}. The paper \cite{identifiers} approaches identifiers via a uniform logical characterization of a certain class of algorithms using $\mathrm{ID}$s, while our work is based on the circuit-style paradigm with formulas and circuits being given based on model size. Thus, the two approaches are not comparable in any uniquely obvious way. Nevertheless, one simple difference between our work and \cite{identifiers} is that we treat $\mathrm{ID}$s bit by bit as concrete bit strings. Thus, we can express, e.g., that the current $\mathrm{ID}$ has a bit $1$ without implying that the current node cannot have the smallest $\mathrm{ID}$ in the system. This is because there is no guarantee on what the set of $\mathrm{ID}$s in the current graph (or distributed system) is, and in a directed graph, we cannot even scan through the graph to find out. On the other hand, the logic in \cite{identifiers} can express, e.g., that the current node has the largest $\mathrm{ID}$, which we cannot do. Of course, with a non-uniform formula allocation function, the circuit-style paradigm can even specify non-computable properties.

The closest work to the current article is \cite{Kuusisto13}, which gives the already mentioned characterization of finite message passing automata via $\mathrm{MSC}$. The paper does not work within the circuit-style paradigm. Furthermore, we cannot turn our circuit to an $\mathrm{FMPA}$ and then use the translation of \cite{Kuusisto13}, as this leads to an exponential blow-up in size. Also, the converse translation is non-polynomial in \cite{Kuusisto13}. 
Furthermore, that paper does not discuss identifiers, or the Cole-Vishkin algorithm, and the work in the paper is based on the paradigm of relating properties directly with single formulae rather than our circuit-style approach.
Concerning further related and \emph{very timely} work, \cite{barcelo} studies graph neural networks (or GNNs) and establishes a match between \emph{aggregate-combine} GNNs and graded modal logic. For further related work on GNNs and logic, see, e.g., \cite{grohe}. Concerning yet further work on logical characterizations of distributed computing models, we mention the theses \cite{tuamothesis}, \cite{reiterthesis}.

\section{Preliminaries}

We let $\Z_+$ denote the set of positive integers. For every $n \in \Z_+$, we let $[n]$ denote the set $\{1,  \ldots, n\}$ and $[n]_0$ the set $\{0,  \ldots, n\}$. 
Let $\mathrm{PROP}$ be a countably infinite set of proposition symbols. We suppose $\mathrm{PROP}$ partitions into two infinite sets $\mathrm{PROP}_0$ and $\mathrm{PROP}_1$, with the intuition that $\mathrm{PROP}_0$ contains ordinary proposition symbols, while $\mathrm{PROP}_1$ consists of distinguished proposition symbols reserved for encoding ID-numbers. We denote finite
sets of proposition symbols 
by $\Pi \subseteq \mathrm{PROP}$. By $\Pi_0$ (respectively, $\Pi_1$), we mean the subset of $\Pi$ containing ordinary (respectively, distinguished) propositions. The set $\mathrm{PROP}$ is associated with a linear order $<^{\mathrm{PROP}}$ which also induces a linear order $<^S$ over
any set $S\subseteq \mathrm{PROP}$.

Let $\Pi$ be a set of proposition symbols.
A \textbf{Kripke model} over $\Pi$ is a structure $(W,R,V)$ with a non-empty domain $W$, an \textbf{accessibility 
relation} $R\subseteq W\times W$ and a 
\textbf{valuation function} $V:\Pi \rightarrow \mathcal{P}(W)$ giving each $p\in \Pi$ a set $V(p)$ of nodes where $p$ is considered true. A
\textbf{pointed Kripke model} is a pair $(M,w)$, 
where $M$ is a Kripke model and $w$ a node in
the domain of $M$. We let $\mathrm{succ}(w)$ denote the set $\{\,v\in W \mid (w, v) \in R\,\}$. 

As in \cite{weak_models, Kuusisto13}, we model distributed systems by Kripke models. An edge $(w,u)\in R$ linking the node $w$ to $u$ via the accessibility relation $R$ means that $w$ can see messages sent by $u$. Thereby, we adopt the convention of \cite{weak_models, Kuusisto13} that messages travel in the direction \emph{opposite} to the edges of $R$. An alternative to this would be to consider modal logics with only inverse modalities, i.e., modalities based on the inverse accessibility relation $R^{-1}$.

Let $k\in \mathbb{N}$ and consider an infinite sequence $S = (\overline{b}_j)_{j\in \mathbb{N}}$ of $k$-bit strings $\overline{b}_j$. Let $A\subseteq [k]$ and $P\subseteq [k]$ be subsets, called \textbf{attention} bits and \textbf{print} bits (or bit positions, strictly speaking). Let $(\overline{a}_j)_{j\in \mathbb{N}}$ and $(\overline{p}_j)_{j\in \mathbb{N}}$ be the corresponding sequences of substrings of the strings in $S$, that is, $(\overline{a}_j)_{j\in \mathbb{N}}$ records the substrings with positions in $A$, and analogously for $(\overline{p}_j)_{j\in \mathbb{N}}$. Let 
$(\overline{r}_j)_{j\in \mathbb{N}}$ be the sequence of substrings with positions in $A\cup P$.
We say that $S$ \textbf{accepts} in round $n$ if 
at least one bit in $\overline{a}_n$ is $1$ and all bits in each $\overline{a}_m$ for $m< n$ are zero. 
Then also $S$ \textbf{outputs} $\overline{p}_n$. 
More precisely, $S$ \textbf{accepts in round $n$ with respect to $(k,A,P)$}, and $\overline{p}_n$ is the \textbf{output of $S$ with respect to $(k,A,P)$}. The sequence $(\overline{r}_j)_{j\in \mathbb{N}}$ is the \textbf{appointed sequence} w.r.t. $(k,A,P)$, 
and the vector $\overline{r}_j$ the \textbf{appointed 
string} of round $j$.


\subsection{Logics}

For a set $\Pi$ of 
proposition symbols, the set of \textbf{$\ml(\Pi)$-formulas} is given by the grammar
\[
    \varphi \Coloneqq \top  \mid  p  \mid \neg \varphi  \mid  ( \varphi \land \varphi )  \mid \Diamond \varphi 
\]
where $p \in \Pi$ and $\top$ is a logical constant symbol. 
The truth of such a formula $\varphi$ in a pointed Kripke model $( M, w ) $ is defined in the ordinary way: $( M, w ) \models p \Leftrightarrow w \in V( p )$, and $( M, w ) \models \Diamond \varphi \Leftrightarrow (M,v) \models \varphi$ for some $v\in W$ such that $(w,v)\in R$. The semantics for $\top,\neg,\wedge$ are the usual ones.

Now, let us fix a set $\mathrm{VAR} \coloneqq \left\{ \, V_i \mid i \in \N \, \right\} $ of \textbf{schema variables}. 
We will mostly use meta variables $X, Y, Z,$ and so on, to 
denote symbols in $\mathrm{VAR}$. The set $\mathrm{VAR}$ is associated with a linear order $<^{\mathrm{VAR}}$ inducing a corresponding linear order $<^{\cT}$ over any $\cT \subseteq \mathrm{VAR}$. Given a set $\cT \subseteq \mathrm{VAR}$ and a set $\Pi\subseteq \mathrm{PROP}$, 
the set of \textbf{$( \Pi, \cT )$-schemata} of \textbf{modal substitution calculus} (or $\msc$) is the set generated by the grammar
\[
    \varphi \Coloneqq \top  \mid  p  \mid V_i  \mid \neg \varphi  \mid  ( \varphi \land \varphi )  \mid \Diamond \varphi 
\]
where $p \in \Pi$ and $V_i \in \cT$. A \textbf{terminal clause} of $\msc$
(over $\Pi$) is a string of the form $V_i( 0 ) \coloneq \varphi$, where $V_i \in \mathrm{VAR}$ and $\varphi \in \mathrm{ML}(\Pi)$. An \textbf{iteration clause} of $\msc$
(over $\Pi$) is a string of the form $V_i \coloneq \psi$ where $V_i \in \mathrm{VAR}$ and $\psi$ is a $( \Pi, \cT ) $-schema for some set $\cT \subseteq \mathrm{VAR}$. In a terminal clause $V_i(0) \coloneq \varphi$,
the symbol $V_i$ is the \textbf{head predicate} and $\varphi$ the \textbf{body} of
the clause. Similarly, $V_i$ is the head predicate of the iteration clause $V_i \coloneq \psi$, 
while $\psi$ is the body.

Let $\cT = \{Y_1, \ldots, Y_k\} \subseteq \mathrm{VAR}$ be a finite, nonempty set of $k$ distinct schema variables. A \textbf{$( \Pi, \cT ) $-program} $\Lambda$ of $\msc$ consists of two lists 
\[
\begin{aligned}
    &Y_1 ( 0 ) \coloneq \varphi_1 \qquad &&Y_1 \coloneq \psi_1 \\ 
    &\vdots &&\vdots \\
    &Y_k ( 0 ) \coloneq \varphi_k &&Y_k \coloneq \psi_k \\ 
\end{aligned}
\] 
of clauses (or rules). The first list contains $k$ terminal clauses over $\Pi$, and
the second contains $k$ iteration clauses whose bodies are $(\Pi,\cT)$-schemata. To fully define $\Lambda$, we also fix two sets $\cP \subseteq \cT$ and $\cA \subseteq \cT$ of \textbf{print predicates} and \textbf{attention predicates} of $\Lambda$. The set $\cP \cup \cA$ is the set of \textbf{appointed predicates} of  $\Lambda$. We call $\Lambda$ a \textbf{$\Pi$-program} if it is a $( \Pi, \cT ) $-program for some $\cT \subseteq \mathrm{VAR}$. The set of head predicates of $\Lambda$ is denoted by $\head(\Lambda)$. 
For each variable $Y_i \in \head ( \Lambda ) $, we let $Y^{0}_i$ denote the body of the terminal clause $Y_i( 0 ) \coloneq \varphi_i$. Recursively, assume we have defined an $\mathrm{ML}( \Pi ) $-formula $Y^{n}_{i}$ for each $Y_i \in \head( \Lambda ) $. Let $\varphi_j$ denote the body of the iteration clause of $Y_j$. The formula $Y^{n+1}_j$ is obtained by replacing each $Y_i$ in $\varphi_j$ with $Y^{n}_i$. Then $Y^n_i$ is the $n$th \textbf{iteration formula of $Y_i$}. Supposing we have fixed a $(\Pi, \cT)$-program $\Lambda$, if $\varphi$ is a $(\Pi, \cT)$-schema, then we let $\varphi^{n+1}$ denote the $\ml(\Pi)$-formula obtained from the schema $\varphi$ by simultaneously replacing each $Y_i \in \head(\Lambda)$ with $Y^n_i$. 
Now, let $( M, w ) $ be a pointed $\Pi$-model and $\Lambda$ a $\Pi$-program of $\msc$. We define that $( M, w ) \models \Lambda $ if for some $n$ and some attention predicate $Y$ of $\Lambda$, we have $( M, w ) \models Y^n$. In Section \ref{extensions}, we will also define output conditions for $\msc$ using print predicates. 
%
%
%

%
%
%
Assume that $p_1, \ldots, p_{\ell}$ enumerate all the distinguished propositions in $\Pi$ in the order $<^{\mathrm{PROP}}$. For each node $w$ of a Kripke model $M$ over $\Pi$, we let $\ID(w)$ denote the \textbf{identifier} of $w$, that is, the $\abs{\Pi_1}$-bit string such that the $i$th bit of $\ID(w)$ is $1$ if and only if $(M,w) \models p_i$. The model $M$ is a Kripke model \textbf{with identifiers} if $\ID(w) \not= \ID(w')$ for each pair of distinct nodes $w$ and $w'$ of $M$. We let $\cK( \Pi, \Delta )$ denote the class of finite Kripke models $(W,R,V)$ over $\Pi$ with identifiers such that the out-degree of each node is at most $\Delta \in \mathbb{N}$. 
For a node $w$, let $s_1,  \ldots, s_d$ be the identifiers of the members of $\mathrm{succ}(w)$ in the lexicographic order. A node $v \in \mathrm{succ}(w)$ is the \textbf{$i$th neighbour} of $w$ iff $\ID(v) = s_i$. 
%
%
%
%
%
%
%
%
%
Analogously to local $\mathrm{ID}$s, if $p_1, \ldots, p_{m}$ enumerate all the propositions in $\Pi$ in the order $<^{\mathrm{PROP}}$,
then the \textbf{local input} of a node $w$ of a Kripke model $M$ over $\Pi$ is the $m$-bit
string $t$ such that the $i$th bit of $t$ is $1$ if and only if $(M,w) \models p_i$.
Now, let $\Lambda$ be a $(\Pi, \cT)$-program of $\msc$, and let $\psi$ be a $(\Pi, \cT)$-schema. We let $\md(\psi)$ denote the \textbf{modal depth} of $\psi$ (i.e., the maximum nesting depth of diamonds $\Diamond$ in $\psi$). 
We let $\mdt(\Lambda)$ (respectively, $\mdi(\Lambda)$) denote the maximum modal depth of the bodies of the terminal clauses (resp., of the iteration clauses) of $\Lambda$. By $\subs(\Lambda)$ we denote the set of all subschemata of $\Lambda$, including head predicates and bodies of iteration and terminal clauses. If $S$ is a set of schemata,  $\subs(S)$ is the set of all subschemata of all schemata in $S$.

\subsection{Circuits and distributed computation}

A \textbf{Boolean circuit} is a directed, acyclic graph where each node of non-zero in-degree is labeled by one of the symbols $\land , \lor , \neg$. The nodes of a circuit are called \textbf{gates}.
The in-degree of a gate $u$ is called the \textbf{fan-in} of $u$, and the out-degree of $u$ is \textbf{fan-out}. The \textbf{input gates} of a circuit are precisely the gates that have zero fan-in; these gates are not labeled by $\land , \lor , \neg$. The \textbf{output-gates} are the ones with fan-out zero; we allow
multiple output gates in a circuit. 
The fan-in of every gate labeled with $\neg$ is $1$. The \textbf{size} $\abs{C} $ of a circuit $C$ is the number of gates in $C$. The \textbf{depth} $d(C)$ of $C$ is the longest path length (number of edges) from an input gate to an output gate. The \textbf{height} $h(G)$ of a gate $G$ in $C$ is the longest path length from an input gate to the gate $G$. Thus, the height of an input gate is zero.
Both the input gates and output gates of a
circuit are linearly ordered. A circuit with $n$ input gates and $k$ output gates then computes a 
function of type $\{0,1\}^n \rightarrow \{0,1\}^k$. This is done in the natural way, analogously to the Boolean operators corresponding to $\wedge,\vee,\neg$; see, for example, \cite{libkin} for the formal definition. The output of the circuit is the \emph{binary string} determined by the output bits of the output gates. Note that 
gates with $\wedge,\vee$ can have any fan-in (also $0$). The $\land$-gates that have zero fan-in always outputs $1$ and therefore corresponds to the $\top$ symbol. The $\lor$-gates that have zero fan-in always outputs $0$ and therefore respectively corresponds to the $\bot$ symbol.

From a Boolean formula, it is easy to define a corresponding circuit by considering its \emph{inverse tree representation}, meaning the tree representation with edges pointing in the inverse direction.
A node $v$ in the inverse tree representation is the \textbf{parent} of $w$ if there is an edge from $w$ to $v$. Then $w$ is a \textbf{child} of $v$. Note that input gates do not have any children and output gates have no parents. The \textbf{descendants} of $w$ are defined such that every child of $w$ is a descendant of $w$, and also every child of a descendant of $w$ is a descendant of $w$.

\begin{definition}
Let $\Pi$ be a set of propositions and $\Delta\in \mathbb{N}$. 
    A \textbf{circuit for $( \Pi, \Delta ) $} is a circuit $C$ that specifies a function $$f:\{0,1\}^{|\Pi|+k(\Delta + 1)} \rightarrow \{0,1\}^{k}$$ 
    for some $k\in \mathbb{N}$. The number $k$ is called the \textbf{state length} of $C$.
    The circuit $C$ is also associated with sets $A\subseteq [k]$ and $P\subseteq [k]$ of \textbf{attention bits} and \textbf{print bits}, respectively. 
    For convenience, we may also call a circuit $C$ for $(\Pi, \Delta)$ a \textbf{message passing circuit} $\mathrm{(}$or $\mpc)$ for $(\Pi,\Delta)$. 

\end{definition}
%
%
%
%
%
%
The set $A\cup P$ is called the set of \textbf{appointed bits} of the circuit.
A circuit $C$ is \textbf{suitable} for a Kripke model $M$ with identifiers if $C$ is a message passing circuit for $(\Pi,\Delta)$, where $\Pi$ is precisely the set of proposition symbols interpreted by $M$, and $\Delta$ is at least the maximum out-degree of the nodes in $M$. A 
circuit $C$ for $(\Pi,\Delta)$ with $|\Pi_1| = m$ is referred to as a \textbf{circuit for $m$ $\mathrm{ID}$-bits}.
We let $\mathrm{CIRC}(\Pi_0,\Delta)$ denote the set of all circuits $C$ such that for some $\Pi$ with $\Pi\cap\mathrm{PROP}_0 = \Pi_0$, the circuit $C$ is a circuit for $(\Pi,\Delta)$. 
We stress that, strictly speaking, when specifying an $\mpc$, we should always specify (together with a circuit) the sets $\Pi$, $\Delta$, the attention and print bits, and an ordering of the input and output gates. 
%
%
%

Before giving a formal definition of distributed computation in a Kripke model $M\in\cK(\Pi,\Delta)$ with a circuit $C$ for $(\Pi,\Delta)$, we describe the process informally.
Each node $u$ of $M$ runs a copy of the circuit $C$. The node $u$ is associated with a local input that is defined as follows. Assume that $p_1, \ldots, p_{\abs{\Pi}}$ enumerate all the propositions in $\Pi$ in the order $<^{\mathrm{PROP}}$. A \emph{local input} at $u$ is $\abs{\Pi}$-bit binary string $\overline{s}$ such that $i$th bit of $\overline{s}$ is $1$ iff $u \in V(p_i)$.  
At the beginning of the computation, the circuit at $u$ reads the string $\overline{s} \cdot 0^{\ell}$ at $u$,
where $\overline{s}$ is the local
input at $u$ and $\ell= k(\Delta + 1)$, so $0^{\ell}$ is simply the part of the input to $C$ that does not
correspond to proposition symbols. Then
the circuit enters a \emph{state}, which is the $k$-bit output string of $C$. Let $s(0,u)$ denote this string; we 
call it the \textbf{state in communication round $0$} at the node $u$. Now, recursively,
suppose we know the state $s(n,u)$ in
communication round $n\in \N$ for
each node $u$. The
state $s(n+1,u)$ for round $n+1$ at $u$ is then computed as follows. 
\begin{enumerate}
\item
At each node $u$, the circuit sends $s(n,u)$ to
the nodes $w$ such that $R(w,u)$. Note here that messages
flow opposite to the direction of $R$-edges. 
\item
The circuit at $u$ \emph{updates} its
state to $s(n+1,u)$,
which is the $k$-bit string obtained as
the output of the circuit with the input
$\overline{s} \cdot \overline{s}_0\cdots \overline{s}_{\Delta}$, which is the
concatenation of the $k$-bit
strings $s_i$ (for $i\in\{0,\dots , \Delta\})$
specified as follows. 
\begin{itemize}
\item 
The string $\overline{s}$ is the local input at $u$.
\item
The string $\overline{s}_0$ is the state $s(n,u)$.
\item
Let $i\in \{1,\dots , m\}$, where $m\leq \Delta$ is the out-degree of $u$. Then $\overline{s}_i$ is the state $s(n,v_i)$ of the $i$th neighbour $v_i$ of $u$. 
%
\item
For $i > m$, we have $\overline{s}_i = 0^k$. 
\end{itemize}
\end{enumerate}

We then define computation of $\mpc$s formally. An $\mpc$ $C$ for $( \Pi, \Delta )$ of state length $k$ and a Kripke model $M = ( W, R, V ) \in \cK( \Pi, \Delta ) $ define a \emph{synchronized distributed system}, which executes an $\omega$-sequence of rounds defined as follows. Each round $n \in \N$ defines a \textbf{global configuration} $f_n \colon W \to \{0,1\}^{k }$. Let $\overline{t}_w$ denote the binary string corresponding to the set of propositions true at $w$ (i.e., local input). The configuration of round $0$ is the function $f_0$ such that $f_0( w )$ is the $k$-bit binary string produced by $C$ with the input $\overline{t}_w \cdot 0^{k(\Delta+1)}$. 
Recursively, assume we have defined $f_{n}$. Let $v_{1}, \ldots, v_{m} \in \mathrm{succ}( w ) $ be the
neighbours of $w$
($m\leq \Delta$)
given in the order of
their $\ID$s. Let $\overline{s}_w$ be the
concatenation $\overline{t}_w \cdot \overline{s}_0 \cdots \overline{s}_{\Delta}$ of $k$-bit binary strings such that 
\begin{enumerate}
\item
$\overline{s}_0 = f_n(w)$, 
\item
$\overline{s}_i = f_n(v_i)$ for 
each $i\in \{1,\dots, m\}$,
\item
$\overline{s}_j = 0^k$ for $j\in \{m+1,\dots , \Delta\}$.
\end{enumerate} 
Then $f_{n+ 1}(w)$ is the 
output string of $C$ with input $\overline{s}_w$. 
Now, consider the sequence $(f_n(w))_{n\in\mathbb{N}}$ of $k$-bit strings that $C$ produces at $w$. Suppose the sequence $(f_n(w))_{n\in\mathbb{N}}$ accepts (resp. outputs $\overline{p}$) in round $n$ with respect to $(k,A,P)$. Then $w$ \textbf{accepts} (resp., \textbf{outputs} $\overline{p}$) in round $n$. Note that the circuit at $w$ keeps executing after round $n$. 

Given a Kripke model $M = (W,R,V)$, a \textbf{solution labeling} is a function $W\rightarrow \{0,1\}^*$ associating nodes with strings. The strings represents outputs of the nodes on distributed computation. 
We could, e.g., label the nodes with strings corresponding to ``yes'' and ``no''. A \textbf{partial solution labeling} for $M$ is a partial function from $W$ to $\{0,1\}^*$, that is, a function of type $U\rightarrow \{0,1\}^*$ for some $U\subseteq W$. Partial solution labelings allow for ``divergent computations'' on some nodes in $W$. The $\textbf{global output}$ of a circuit $C$ over a model $M = (W,R,V)$ is a
function $g:U\rightarrow\{0,1\}^*$ such that 
\begin{enumerate}
    \item 
    $U\subseteq W$,
    \item
    for all $w\in U$, the circuit $C$ outputs $g(w)$ in some round $n$, and
    \item
    $C$ does not produce an output for any $v\in W\setminus U$.
\end{enumerate}

Now, fix a finite set $\Pi_0\subseteq \mathrm{PROP}_0$ of proposition symbols. Intuitively, these are the ``actual'' propositions in models, while the set of $\mathrm{ID}$-propositions will grow with model size. Let $\mathcal{M}(\Pi_0)$ denote the class of all finite Kripke models $M$ with $\mathrm{ID}$s and having a set $\Pi$ of proposition symbols such that $\Pi\cap\mathrm{PROP}_0 = \Pi_0$. Thus, $\Pi_0$ is the same for all models in $\mathcal{M}(\Pi_0)$ but the symbols for $\mathrm{ID}$s vary. Consider a subclass $\mathcal{M}\subseteq \mathcal{M}(\Pi_0)$. Now, a distributed computing \textbf{problem} over $\mathcal{M}$ is a mapping $p$ with domain $\mathcal{M}$ that associates with each input $M$ a (possibly infinite) set $p(M)$ of partial solution labelings for $M$. The set $p(M)$ represents the set of acceptable answers to the problem $p$ over $M$. Many graph problems (e.g., colourings) naturally
involve a set of such answer labelings.

For $\Delta\in \mathbb{N}$, we let $\mathcal{M}(\Pi_0,\Delta)$ denote the restriction of $\mathcal{M}(\Pi_0)$ to models with maximum out-degree $\Delta$.
A \textbf{circuit sequence}
for $\mathcal{M}(\Pi_0,\Delta)$ is a function $F : \mathbb{Z}_+\to \mathrm{CIRC}(\Pi_0,\Delta)$ such that $F(n)$ is a circuit for $\lceil\log n\rceil$ $\mathrm{ID}$-bits. 
%
%
%
%
%
%
Now, $F$ \textbf{solves} a problem $p$ over $\mathcal{M}(\Pi_0,\Delta)$ if the global output of $F(n)$ belongs to $p(M)$ for each $M\in \mathcal{M}(\Pi_0,\Delta)$ of domain size $n$. Let $c\in \mathbb{N}$. 
We define $\mathrm{DCC}_{\Delta}^c[\log n]$ to be the class of distributed computing problems solvable by a circuit sequence $F$ of maximum fan-in $c$ circuits such that 
the size of $F(n)$ is $\mathcal{O}(\log n)$. The related \textsc{LogSpace} uniform class requires that
each $F$ can be computed in $\textsc{LogSpace}$.  $\mathrm{DCC}$ stands for \emph{distributed
computing by circuits}.

\section{Extending MSC}\label{extensions}

Here we define some extensions of $\msc$ to be used mainly as tools in the proofs that follow. 
%
%
%
Let $\Pi$ be a set of propositions and $\cT$ a set of schema variables.
Let $\Delta\in \mathbb{N}$. 
In \textbf{Multimodal $\msc$} (or $\mathrm{MMSC}$), instead of $\Diamond$, we have the operators $\Diamond_1,\dots , \Diamond_{\Delta}$, and otherwise the syntax is as in $\msc$. The schema $\Diamond_i \varphi$ simply asks if $\varphi$ is true at the $i$th neighbour. More formally, if $(M, w)$ is a pointed Kripke model with identifiers, then
$(M, w) \models \Diamond_i \varphi \Leftrightarrow (M, v_i) \models \varphi$ such that $(w, v_i) \in R$  and $v_i$ is the $i$th neighbour of $w$, noting that
if the out-degree of $w$ is lesser than $i$, then $\Diamond_i \varphi$ is false at $w$. 
A \textbf{$(\Pi, \Delta)$-program} of $\mmsc$ is exactly like a $\Pi$-program of $\msc$ but we are only allowed to use operators $\Diamond_1, \ldots, \Diamond_{\Delta}$ instead of $\Diamond$. A \textbf{$\Pi$-program} $\Lambda$ of $\mmsc$ is a $(\Pi, \Delta)$-program for any $\Delta \geq d$, where $d$ is the maximum subindex in any diamond in $\Lambda$. We also fix print and attention predicates for programs of $\mmsc$. 
Note that $\mathrm{MMSC}$ is not a logic in the usual sense as the
operators $\Diamond_i$ require information about
the predicates defining $\ID$s. This could be remedied via signature changes and limiting attention to multimodal models with relations having out-degree at most one. This would be a slightly messy approach, and the current approach suffices for this article.

We then define \textbf{$\msc$ with conditional rules} (or $\mathrm{CMSC}$). Here, we allow ``if-else'' rules as iteration clauses. Let $\varphi_1, \ldots, \varphi_{n}$ and $\psi_1, \ldots, \psi_{n}$ and also $\chi$ be $(\Pi, \cT)$-schemata of basic $\msc$. A \textbf{conditional iteration clause} is a rule of the form 
$$X \coloneq_{\varphi_1, \dots, \varphi_n} \psi_1; \dots; \psi_{n}; \chi .$$
The schemata $\varphi_i$ are \textbf{conditions} for the head predicate $X$ and the schemata $\psi_i$ are the related \textbf{consequences}. The last schema $\chi$ is called the \textbf{backup}. Note that when $n = 0$, we have a standard $\msc$ clause.  
\textbf{$\Pi$-programs} of $\mathrm{CMSC}$ are exactly as for $\mathrm{MSC}$, but 
we are allowed to use conditional iteration clauses. Thus a program $\Lambda$ of $\mathrm{CMSC}$ consists of $k$ terminal clauses, $k' \leq k$ conditional iteration clauses and $k-k'$ standard iteration clauses for some $k\in\mathbb{Z}_+$.
Again we also fix some sets of schema
variables as print and attention predicates. 

 To define the semantics, we will specify---as in $\msc$---the $n$th iteration formula of each head predicate. Informally, we always use the first (from the left) condition $\varphi_i$ that holds and thus evaluate the corresponding consequence $\psi_i$ as the body of our rule. If none of the conditions hold, then we use the backup. 
 Let $\Lambda$ be a $\Pi$-program of $\mathrm{CMSC}$.
First, we let
the zeroth iteration clause $Y_i^{0}$ of a
head predicate $Y_i \in \head(\Lambda)$ be the 
the terminal clause of $Y_i$. 
Recursively, assume we have defined an $\mathrm{ML}( \Pi ) $-formula $Y^{n}_{i}$ for each $Y_i \in \head( \Lambda ) $. 
Now, consider the rule 
$$
Y_i \coloneq_{\varphi_1, \dots, \varphi_m} \psi_1; \dots; \psi_m; \chi.
$$
%
%
%
Let $\varphi_j^{n+1}$ be the formula 
obtained by replacing each 
schema variable $Y_k$ in
the condition $\varphi_j$ with $Y_k^n$. 
The formulae $\chi^{n+1}$ and $\psi_k^{n+1}$ are obtained analogously.
Then, the formula $Y^{n+1}_i$ is 
\[
\bigvee\limits_{k\leq m}
\Big(\Big(\bigwedge\limits_{j < k} \neg \varphi_j^{n+1}\Big) \land \varphi_k^{n+1} \land \psi_k^{n+1}\Big) \ \lor \ \ \Big(\Big(\bigwedge\limits_{j \leq m} \neg \varphi_j^{n+1}\Big) \land \chi^{n+1}\Big).
\]

Often the backup 
schema $\chi$ is just the head predicate $X$ of the rule. This means that
the truth value of the head predicate does not change if none of the conditions hold.
We say that a condition $\varphi_k$ is \textbf{hot} at $w$ in round $n\geq 1$ if the formula $\varphi_k^n$ is true at $w$ and none of the ``earlier'' conditions $\varphi_j^n$ of the same rule (so $j<k$) are true. Otherwise, the backup is hot. 
We say that a conditional iteration clause (or the corresponding head predicate) is \textbf{active} in round $n \geq 1$ at node $w$ if one of the condition formulas of the rule is hot.

Last, we define \textbf{message passing $\msc$} (or $\mathrm{MPMSC}$) roughly as multimodal $\msc$ with conditional rules. The \textbf{$(\Pi, \Delta)$-programs} are exactly like $(\Pi,\Delta)$-programs of $\mmsc$ with conditional rules and the following restrictions.
\begin{enumerate}
\item
The modal depths of terminal clauses and conditions of rules are zero.
\item 
The consequences, backups and bodies of standard iteration clauses all have modal
depth at most one. 
\end{enumerate}
As in $\mmsc$, operators $\Diamond$ are not allowed. A $\Pi$-program of $\mpmsc$ is defined analogously to a $\Pi$-program of $\mmsc$. Thus, a program of $\mpmsc$ contains $k$ terminal clauses, $k' \leq k$ conditional iteration clauses and $k-k'$ standard iteration clauses for some $k \in Z_+$. We also fix sets of attention and print predicates. The semantics are defined as for $\cmsc$, noting that now diamonds $\Diamond_i$ are used. A non-terminal clause of a program of $\mpmsc$ is a \textbf{communication clause} if it contains at least one diamond. A communication clause is \textbf{broadcasting} in round $n \in \mathbb{Z}_+$ at $w$ if one of the following holds.
\begin{enumerate}
    \item 
    A condition $\varphi_i$ is hot at $w$ and the corresponding consequence has a diamond.
    \item
    A backup is hot at $w$ and has a diamond.
    \item
    The rule is not conditional but has a diamond.
\end{enumerate}

We will next define acceptance and output conditions for programs of all variants of $\msc$, including standard $\msc$. The acceptance conditions will be consistent with the already given conditions for standard $\msc$.
Let $\Lambda$ be a program and
$\cA$ and $\cP$ the sets of attention and 
print predicates. Let $Y_1, \ldots, Y_k$ enumerate the head predicates in $\Lambda$ in the order $<^{\mathrm{VAR}}$. Let $M = (W, R, V)$ be a Kripke model. Each round $n \in \N$ defines a \textbf{global configuration} $g_n \colon W \to \{0,1\}^k$ given as follows. The configuration of the $n$th round is the function $g_n$ such that the $i$th bit of $g_n(w)$ is $1$ if and only if $(M, w) \models Y^n_i$. If the sequence $(g_n(w))_{n \in \N}$ accepts (respectively outputs $\overline{p}$) in round $n$ with respect to $(k, \cA, \cP)$, then we say that the node $w$ \textbf{accepts} (respectively \textbf{outputs} $\overline{p}$) in round $n$. Then $n$ is the \textbf{output round} of $\Lambda$ at $w$. We write $(M, w) \models \Lambda$ if node $w$ accepts in some round $n$. 
%
%
%
For a program $\Lambda$ of message passing $\msc$ and model $M$, a \textbf{global communication round} is a computation round $n$ where at least one communication clause is broadcasting in at least one node of $M$. A program $\Lambda$ \textbf{outputs $\overline{p}$ at $w$ in global communication time} $m$ if the output round of $\Lambda$ at $w$ is $n$
and $m\leq n$ is the number of
global communication rounds in
the set $\{0,\dots , n\}$ of rounds.


Now, let $\mathcal{L}$ denote the set of all programs of all of our variants of $\msc$.
Let $\mathcal{C}$ denote the set of all $\mpc$s. 
For each $\Lambda\in\mathcal{L}$, we say that a Kripke model $M$ is \textbf{suitable} for $\Lambda$ if $M$ interprets (at least) all the proposition symbols that occur in $\Lambda$. For a message passing circuit for $(\Pi, \Delta)$, we say that $M$ is \textbf{suitable} for the circuit if the set of proposition symbols interpreted by $M$ is precisely $\Pi$ and the maximum out-degree of $M$ is at most $\Delta$.
Now, let $x$ and $y$ be
any members of $\mathcal{C}\cup \mathcal{L}$.
We say that $x$ and $y$ are (acceptance) \textbf{equivalent} if for each Kripke model $M$ that is suitable for both $x$ and $y$ and for each node $w$ in in the model, $x$ and $y$ produce the same output at $w$ or neither produce any output at all at $w$.
We say that $x$ and $y$ are \textbf{strongly equivalent} if for 
each $M$ suitable for $x$ and $y$ and for each node $w$ in the model and in every round $n$, the objects $x$ and $y$ produce the same appointed string $\overline{r}_n$ at $w$. We also define a special equivalence notion for $\mpmsc$ and $\mpc$. We say that a program $\Lambda$ of $\mpmsc$ and circuit $C$ are \textbf{strongly communication equivalent} if for each $M$ suitable for both $\Lambda$ and $C$ and for each node $w$ in the model, the appointed sequence $S$ of the circuit is precisely the sequence $(\overline{r}_j)_{j\in \{0\}\cup G}$ of appointed strings of the program, where $G\subseteq\mathbb{Z}_+$ is the set of global communication rounds $n$ of the program. 
%
%
%
%
%
%
Finally, the \textbf{length} of a program of any variant of $\msc$ is the number of \emph{occurrences} of proposition symbols, head predicates, and operators $\top$, $\neg$, $\wedge$, $\Diamond$, $\Diamond_i$.
The modal depth $\md(\Lambda)$ of a program $\Lambda$ is the maximum modal depth of its rule bodies (iteration
and terminal).

\section{Linking MPMSC to message passing circuits}\label{Descriptive characterizations}


To obtain the desired descriptive characterizations, we begin by translating $\mpc$s to  $\mpmsc$.

\subsection{From MPC to MPMSC}\label{sec:mpc_to_mpmsc}

To ultimately translate $\mpc$s to $\mpmsc$, we will first show how to simulate the evaluation of a standard Boolean circuit with a diamond-free program of $\msc$.
Let $C$ be a circuit of depth $d$ with $\ell$ input and $k$ output gates. Let $L$ denote any of the variants of $\msc$. Fix schema variables $I_1,\dots I_{\ell}$ and $O_1,\dots, O_k$, with both sequences given here in the order $<^{\mathrm{VAR}}$.
Consider a program $\Lambda$ of $L$ with the following properties.
\begin{enumerate}
\item
The set of schema variables of $\Lambda$ contains (at least)
the variables $I_1,\dots I_{\ell}, O_1,\dots
, O_k$. 
\item
The program has no diamond operators ($\Diamond$
or $\Diamond_i$) and contains no proposition symbols.
\item
The terminal clause for
each schema variable $X$ is $X(0) \coloneq\ \bot$.
%
%
%
%
%
%
%
%
%
%
\end{enumerate}
Let $P\colon\{\bot,\top\}^{\ell}\rightarrow \{\bot,\top\}^k$ be
the function defined as follows. For each input
$(x_1,\dots , x_{\ell}) \in \{\bot,\top\}^{\ell}$ to $P$, modify $\Lambda$ to a new program $\Lambda(x_1,\dots , x_{\ell})$ by changing each terminal clause $I_i(0) \coloneq \bot$ to
$I_i(0)\coloneq x_i$. Let $(y_1,\ldots y_k)\in \{\bot,\top\}^k$ be the tuple of truth values of the $d$th iteration formulas $O^d_1,
\ldots , O^d_k$, where we recall that $d$ is the depth of our circuit $C$. Then we define $P(x_1,\ldots , x_{\ell}) := (y_1,\ldots , y_k)$.
Now, if $P$ defined this way is
identical to the function computed by $C$, then $\Lambda$ \textbf{simulates}
the circuit $C$ (w.r.t. $I_1,\dots, I_{\ell}$
and $O_1,\dots , O_k$).

\begin{lemma}\label{circuitsimulationlemma} 
For each circuit $C$ of size $m$ 
and constant fan-in, there exists a
program of $L$ of size $\ordo(m)$ that simulates $C$, 
where $L$ is any of the variants of $\msc$. 
\end{lemma}
\begin{proof}
Assume first that the depth $d$ of $C$ is at least $1$. 
We modify $C$ so that the height of each output gate is 
precisely $d$ by using identity gates, that is, $\wedge$-gates
with fan-in $1$.
Then we define a schema variable for each gate of the obtained circuit $C'$. The variables of the input gates are $I_1,\dots , I_{\ell}$ while those of the output gates 
are $O_1,\dots , O_{k}$.
Let $X$ be a schema variable 
for a $\wedge$-gate $G$ of $C'$.
 We define a corresponding terminal clause $X(0) \coloneq \bot$ and iteration clause $X \coloneq Y_1\wedge \cdots \wedge Y_j$,
where $Y_1,\dots , Y_j$ are the variables for the gates that connect to $G$. Similarly, for a variable $X'$ for a disjunction gate $G$', we define the rules $X'(0) \coloneq \bot$ and $X' \coloneq Y_1 '\vee \cdots \vee Y_j'$ where $Y_1',\dots , Y_j'$ are the variables for the gates connecting to $G'$. For negation, we define $X''(0) \coloneq \bot$ and $X'' \coloneq \neg Y$, where $Y$ is the variable for the connecting gate. 
We let the terminal clauses for the head predicates $I_i$ relating to input gates be $I_i(0) \coloneq \bot$. This choice is irrelevant, as when checking if a program simulates a circuit, we modify the terminal rules to match input strings. The related iteration clause 
is $I_i \coloneq I_i$.

Finally, in the extreme case where the depth of $C$ is $0$, so each 
input gate is also an output gate, we define the program
with the head predicate
sequence $(I_{1},\dots , I_{\ell}) = (O_{1},\dots , O_{k})$
and such that the (terminal and 
iteration) clause for each head predicate $I_i = O_i$
is $I_i \coloneq \bot$.
\end{proof}

\begin{theorem}\label{thr:mpc_to_mpmsc}
Given an $\mpc$ for $(\Pi, \Delta)$ of size $m$, we can construct a communication equivalent $(\Pi,\Delta)$-program of $\mpmsc$. Supposing a constant bound $c$ for the fan-in of\, $\mpc$s, the size of the program is linear in the size of the circuit. 
\end{theorem}

\begin{proof}
    Let $C$ be an $\mpc$ for $(\Pi, \Delta)$ of state length $k$. We will first explain informally how our program $\Lambda_C$ for circuit $C$ will work. The program $\Lambda_C$ uses $k$ head predicates to simulate the state of the circuit. We will use Lemma \ref{circuitsimulationlemma} to build our program, and the operators $\Diamond_i$ will be used to simulate receiving messages of neighbours. The program $\Lambda_C$ computes in repeated periods of $d+1$ rounds, where $d = d(C)$ is the depth of $C$. Simulating the reception of neighbours' messages takes one round, and the remaining $d$ rounds go to simulating the evaluation of the circuit.

We then present the formal proof. 
First, we define a simple clock of length $d(C) + 1$ that will be used for timing the program correctly. The clock consists of the head predicates $T_0, T_1, \ldots, T_{d(C)}$ and the following rules: $T_0 (0) \coloneq \bot$, $T_0 \coloneq T_{d(C)}$, $T_1 (0) \coloneq \top$, $T_1 \coloneq T_0$, and for $i \in [d(C)-1]$, we have $T_{i+1} (0) \coloneq \bot$ and $T_{i+1} \coloneq T_i$. In every round, precisely one of the head predicates $T_i$ is true and the others are false. In round $0$, the only true predicate is $T_1$, and in round $i \in [d(C)-1]$, the only true predicate is $T_{i+1}$. After $d(C)$ rounds, $T_0$ is true, and in the next round the clock starts over again.

Let $\Gamma_C$ be a program simulating the internal evaluation of the circuit $C$ as given in the proof of Lemma \ref{circuitsimulationlemma}. We will rewrite some of the iteration clauses of $\Gamma_C$ as follows. If $X_G$ is a head predicate corresponding to a non-input gate $G$ in $\Gamma_C$, then we rewrite the corresponding iteration clause $X_G \coloneq \varphi$ to $X_G \coloneq_{T_{h(G)}} \varphi; X_G$, where $h(G)$ is the height of the gate $G$.

For every $\ell \in [\, \abs{\Pi}\, ]$, we let $I_{\ell}^{\Pi}$ refer to the head predicate of $\Gamma_C$ that corresponds to the input gate of $C$ that reads the truth value of proposition $p_{\ell}$. For every $i \in [k]$ and $j \in [\Delta]_{0}$, we let $I_{(i,j)}$ refer to a head predicate of $\Gamma_C$ that corresponds to the input gate of $C$ that reads the $i$th value of the state string of the $j$th neighbour. The ``neighbour $0$'' refers to the home node. Next, we will rewrite the clauses with head predicates corresponding to input gates. For every $i \in [k]$, we let $O_i$ refer to the head predicate of $\Gamma_C$ that corresponds to the $i$th output gate of $C$. The terminal (respectively, iteration) clause for $I^{\Pi}_i$ is rewritten to be $I^{\Pi}_i(0) \coloneq p_i$ (resp., $I^{\Pi}_i \coloneq_{T_0} p_i; I^{\Pi}_i$). If $j \neq 0$, then the terminal (resp., iteration) clause for every $I_{(i,j)}$ is rewritten to be $I_{(i,j)}(0) \coloneq \bot$ (resp., $I_{(i,j)} \coloneq_{T_0} \Diamond_j O_i; I_{(i,j)}$). The terminal (resp., iteration) clause for every $I_{(i, 0)}$ is rewritten to be $I_{(i,0)}(0) \coloneq \bot$ (resp., $I_{(i,0)} \coloneq_{T_0} O_i; I_{(i,0)}$).

The attention and print predicates of $\Lambda_C$ are defined as follows. Let $A \subseteq [k]$ (resp., $P \subseteq [k]$) be the set of the attention (resp., print) bit positions in $C$. The attention (resp., print) predicates of $\Lambda_C$ are precisely the head predicates $O_j$, where $j \in A$ (resp., $j \in P$).

We analyze how $\Lambda_C$ works. The program executes in a periodic fashion in cycles with $d(C) + 1$ rounds in each cycle.
In round $0$, the program $\Lambda_C$ reads the proposition symbols and records the local input with the head predicates $I_i^{\Pi}$ whose truth values will remain constant for the rest of the computation. 
Also, $T_1$ becomes true in round $0$. 
In round $1$, the head predicates corresponding to gates at height one are active and thus updated. (Note that the predicates $I_{(i,j)}$ for input gates are inactive because $T_0$ is false. They stay false in round $1$ because they become false in round $0$ and the backup does not affect the truth value.)
From height one, the execution then continues to predicates for gates at height two, and so on. In round $d(C)$, the head predicates for output gates $O_i$ are active. The program also prints if an attention predicate is true. In round $d(C) + 1$, the predicate $T_0$ is true. Thus, the input gate predicates $I_{(i,j)}$ are active, and the cycle starts again as they update using diamonds $\Diamond_i$. We obtain truth values that correspond to an input string to our circuit. We then continue by simulating height one in round $d(C) +2$, continuing in further rounds all the way up to height $d(C)$ gates and finishing the second cycle of the execution of $\Lambda_C$. The subsequent cycles are analogous. Thus, our program $\Lambda_C$ simulates the circuit $C$ in a periodic fashion.

The program $\Lambda_C$ is equivalent and communication equivalent to $C$. The communication clauses in $\Lambda_C$ are synchronous, i.e., every node broadcasts in the same rounds. This is because simulating the circuit takes the same amount of time at every node. Thus, the output and global communication times of $\Lambda_C$ match with $C$. The translation is clearly linear in the size of $C$ (for constant fan-in $C$) due to Lemma \ref{circuitsimulationlemma}. 
\end{proof}

\subsection{From MPMSC to MPC}\label{sec:mpmsc_to_mpc}

Converting an $\mpmsc$-program to a circuit is, 
perhaps, easier. The state string of the constructed $\mpc$ essentially stores the values of the head predicates and proposition symbols used by the program and computes a new state string by simulating the program clauses. We begin with the following lemma that shows how to get rid of conditional rules. 
The proof---given in the appendix---is based on expressing the conditions of conditional clauses within a standard clause. The non-trivial part is to keep the translation linear. 

\begin{lemma}\label{lemma:cmsc_to_msc}
 Given a $\Pi$-program of $\cmsc$, we can construct a strongly equivalent $\Pi$-program of $\msc$ of size linear in the size of the $\cmsc$-program and with the same maximum modal depth in relation to both the terminal clauses and the iteration clauses. 
\end{lemma}

It is easy to get the following corresponding result for $\mpmsc$ from the proof of the previous lemma, recalling that terminal clauses in $\mpmsc$ are always of modal depth zero.

\begin{corollary}\label{lemma:mpmsc_to_mmsc}
    Given a 
    $\Pi$-program of $\mpmsc$ of size $m$, we can construct a strongly equivalent $\Pi$-program of $\mmsc$ of size $\ordo(m)$ and with the same maximum modal depth of iteration clauses and with terminal clauses of modal depth zero. All diamond operators in the constructed program also appear in the original one. 
\end{corollary}

Let $d \in \mathbb{Z}_+$. Let $x$ be either a proposition
symbol or a head predicate.
The symbol $x$ is \textbf{$d$-omnipresent} in a program $\Lambda$ of $\mpmsc$ if the following
conditions hold.
\begin{enumerate}
    \item 
    The symbol $x$ appears in some body of an iteration clause such
    that it does \emph{not} appear in the scope of
    any diamond.
    \item
For each $i\in [d]$, the symbol $x$ appears in some body of an
iteration clause within the scope of $\Diamond_i$. 
\end{enumerate}

A proposition symbol $p$ is \textbf{weakly $d$-omnipresent} in $\Lambda$ if the \emph{disjunction} of the following conditions holds.
\begin{enumerate}
    \item
    The proposition $p$ does not appear in any iteration clause of $\Lambda$.
    \item
    The proposition symbol $p$ is $d$-omnipresent in $\Lambda$.
\end{enumerate}

Before giving a translation from an $\mpmsc$-program to $\mpc$, we prove the following lemma.

\begin{lemma}\label{lem:perform}
    Given a $\Pi$-program $\Lambda$ of $\mpmsc$ and $d \in \Z_+$, we can construct a strongly equivalent $\Pi$-program of $\mpmsc$ with the following conditions.
    \begin{enumerate}
        \item Every proposition symbol in $\Pi$ appears in the program.
        \item Every proposition symbol in $\Lambda$ is weakly $d$-omnipresent.
        \item Every head predicate of $\Lambda$ is $d$-omnipresent.
        \item The size of the program is $\ordo(d\,\abs{\Lambda} + \abs{\Pi})$.
    \end{enumerate}
\end{lemma}

\begin{proof}
For each $p \in \Pi$ that does \emph{not} appear in $\Lambda$, we add $p$ to the program by replacing some body $\psi$ of a 
terminal clause by $\psi \land (p \lor \neg p)$. Thus, we ultimately get a program where each $p\in \Pi$ appears.

For each head predicate $Y$ that is not $d$-omnipresent in $\Lambda$, we modify the program by replacing some schema $\psi$ of a non-conditional iteration clause (or the backup of a conditional iteration clause) by 
$$\psi \land (Y \lor \neg Y) \land \bigwedge_{i \in [d]} \Diamond_i (Y \lor \neg Y).$$

Let $q$ be a proposition symbol in $\Lambda$ that is not weakly $d$-omnipresent in $\Lambda$. If $q$ does appear in an iteration clause of the program, then we make it $d$-omnipresent in the same way as we did for head predicates. If $q$ does not appear in an iteration clause of the program, then we do nothing.

Now, our program is ready. Clearly, it is strongly equivalent to $\Lambda$ and has the size \[\ordo((d+1)\, \abs{\Lambda} + \abs{\Pi}) = \ordo(d\,\abs{\Lambda} + \abs{\Pi}),\]   
as wanted.
\end{proof}

Next, we prove another lemma that is used as a tool to translate an $\mpmsc$-program into an $\mpc$-program. The lemma can be used to simulate two circuits with one.

\begin{lemma}\label{lem:two_circuits}
    Let $C_0$ and $C_1$ be circuits with $k_0$ and $k_1$ input gates and with the same number of output gates. From $C_0$ and $C_1$, we can construct a circuit $C$ with $k_0 + k_1 + 1$ input gates with the following properties.
    Let $\overline{s}_0$ be a $k_0$-bit input string and $\overline{s}_1$ a $k_1$-bit input string. Let $\overline{s} = \overline{s}_0 b \overline{s}_1$, where $b \in \{0,1\}$. 
    \begin{enumerate}
        \item If $b = 0$, and $\overline{t}_0$ is the output of $C_0$ with input $\overline{s}_0$, then $C$ outputs $1\, \overline{t}_0$ with input $\overline{s}$.
        \item If $b = 1$, and $\overline{t}_1$ is the output of $C_1$ with input $\overline{s}_1$, then $C$ outputs $1\, \overline{t}_1$ with input $\overline{s}$.
        \item The first $k_0$ input gates of $C$ are equivalent to the input gates of $C_0$ and the last $k_1$ input gates of $C$ are equivalent to the input gates of $C_1$.
    \end{enumerate}
\end{lemma}

\begin{proof}
    Before giving a formal construction, we give an example (drawn below) where $C_0$ has $3$ input gates, $C_1$ has $5$ input gates, and both have $2$ output gates (drawn in gray). In the figure, the symbols $C_0$ and $C_1$ correspond to the circuits, and only the input and output gates are drawn to illustrate the construction. The blue $\lor$-gate (drawn in a double circle) is the output gate that outputs the bit $1$ which begins the output string of $C$. The red $\lor$-gates at the top of the picture are new output gates that mimic the output of $C_0$ if $G_{\text{in}}$ gets $0$ as input; otherwise, the red $\lor$-gates mimic the output of $C_1$. The input gate $G_{\mathit{in}}$ is the one that takes the bit $b$ in the input $\overline{s}_0b\overline{s}_1$ to $C$.

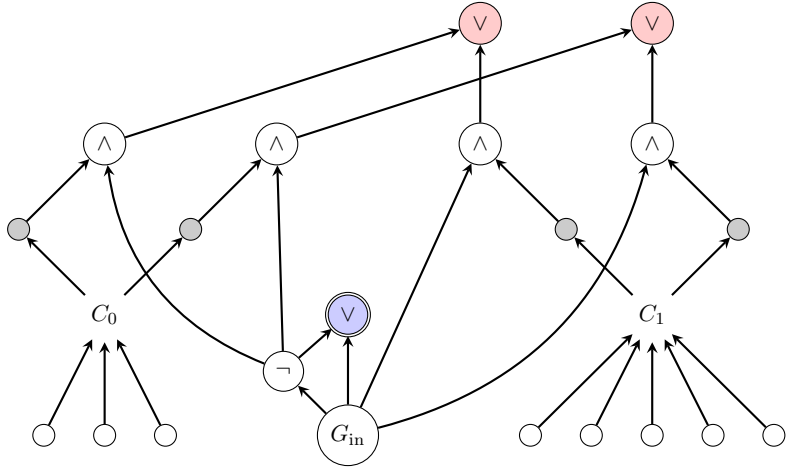
    \begin{center}
    \begin{tikzpicture}[scale=0.8, every node/.style={scale=0.8}, nodes={draw, circle}, <-]
    \node(I1) {};
    \node[right of=I1](I2) {}; 
    \node[right of=I2](I3){};

    \node[right of=I3, node distance=3cm](G) {$G_{\text{in}}$};

    \node[right of=G, node distance=3cm](i1) {};
    \node[right of=i1](i2) {}; 
    \node[right of=i2](i3){};
    \node[right of=i3](i4){};
    \node[right of=i4](i5){};

    \node[above of=I2, node distance=2cm, draw=none](C0) {$C_0$};
    \node[above of=i3, node distance=2cm, draw=none](C1) {$C_1$};
    \node[above left of=G, node distance=1.5cm](N1) {$\neg$};
    
    \node[above left of=C0, node distance=2cm, fill=black!20](O1) {};
    \node[above right of=C0, node distance=2cm, fill=black!20](O2) {};
    \node[above left of=C1, node distance=2cm, fill=black!20](o1) {};
    \node[above right of=C1, node distance=2cm, fill=black!20](o2) {};

    \node[above left of=o1, node distance=2cm](l1) {$\land$};
    \node[above right of=O1, node distance=2cm] (L1) {$\land$};
    
    \node[above left of=o2, node distance=2cm](l2) {$\land$};
    \node[above right of=O2, node distance=2cm] (L2) {$\land$};

    \node[above of=l1, node distance=2cm, fill=red!20] (d1) {$\lor$};
    \node[above of=l2, node distance=2cm, fill={red!20}] (d2) {$\lor$};
    \node[above of=G, node distance=2cm, fill=blue!20, double](D1) {$\lor$};

    \path [-stealth, thick]
    (I1) edge node [draw=none] {} (C0)
    (I2) edge node [draw=none] {} (C0)
    (I3) edge node [draw=none] {} (C0)
    (i1) edge node [draw=none] {} (C1)
    (i2) edge node [draw=none] {} (C1)
    (i3) edge node [draw=none] {} (C1)
    (i4) edge node [draw=none] {} (C1)
    (i5) edge node [draw=none] {} (C1)
    
    (G) edge node [draw=none] {} (N1)
    (G) edge node [draw=none] {} (D1)
    (N1) edge node [draw=none] {} (D1)

    (C0) edge node [draw=none] {} (O1)
    (C0) edge node [draw=none] {} (O2)

    (C1) edge node [draw=none] {} (o1)
    (C1) edge node [draw=none] {} (o2)

    (O1) edge node [draw=none] {} (L1)
    (O2) edge node [draw=none] {} (L2)
    (N1) edge [bend left] node [draw=none] {} (L1)
    (N1) edge node [draw=none] {} (L2)
    
    (o1) edge node [draw=none] {} (l1)
    (o2) edge node [draw=none] {} (l2)
    (G) edge node [draw=none] {} (l1)
    (G) edge [bend right] node [draw=none] {} (l2)

    (L1) edge node [draw=none] {} (d1)
    (l1) edge node [draw=none] {} (d1)
    (l2) edge node [draw=none] {} (d2)
    (L2) edge node [draw=none] {} (d2);

    \end{tikzpicture}    
\end{center}

    To construct $C$, we first construct a circuit $C'$ for the fresh input gate $G_{\text{in}}$ such that $C'$ outputs $1$ no matter what $G_{\text{in}}$ gets as input. This circuit can be constructed by using a $\neg$-gate and a $\lor$-gate---which is the blue output gate in a double circle---as in the example. We use $C'$ to output the bit $1$ in the beginning of the output of $C$, and we use $G_{\text{in}}$ and the $\neg\,$-gate as flag for $C$ to tell which circuit ($C_0$ or $C_1$) we should use as output.  
    
    The rest of the circuit $C$ is constructed as follows. We copy $C_0$ and $C_1$ to $C$ and also add the circuit $C'$ to $C$. The inputs of the circuit $C$ is ordered such that we have the inputs of $C_0$ on the left, those of $C'$ in the middle, and those of $C_1$ on the right. Next, we connect the output gates of the circuits $C_0$ and $C_1$ by using $G_{\text{in}}$ and the $\neg\,$-gate of $C'$ as follows. If $G_0$ is the $i$th output gate of $C_0$, and respectively, $G_1$ is the $i$th output gate of $C_1$ (recall that circuits have the same number of output gates), then we can construct a fresh output gate that outputs the value of $G_0$ if and only if $G_{\text{in}}$ gets $0$ as input, and otherwise it outputs the value of $G_1$. This can be constructed as follows.
\begin{enumerate}
    \item 
    For each output-gate of $C_0$, we introduce a fresh $\land$-gate,
    and we then connect each of the output-gates of $C_0$ bijectively
    to these $\land$-gates. We also similarly introduce a fresh $\land$-gates for each output-gate of $C_1$ and link them in a similar fashion. 
\item
We link the $\neg\,$-gate of $C'$ to the fresh $\land$-gates of $C_0$.
\item
We link the $G_{\mathit{in}}$ to the fresh $\land$-gates of $C_1$.
\item
We finally introduce fresh $\lor$-gates $G_1,\dots , G_p$, one for
each output of $C_0$.
Then we link the fresh $\land$-gates for $C_0$ bijectively to these $\lor$-gates. Finally, we similarly link the fresh $\land$-gates of $C_1$ to 
the gates $G_1, \dots , G_p$ (see the figure), making sure that the order of the output strings of of $C_0$ and $C_1$ is correctly simulated. 
\end{enumerate}
%
%
%
Now, clearly our circuit $C$ works as wanted in the statement.
\end{proof}

We are now ready to prove the following. 

\begin{theorem}\label{thr:mpmsc_to_mpc}
Given $\Pi$, $\Delta$ and a $\Pi$-program of $\mpmsc$ of size $m$, we can construct a strongly equivalent $\mpc$ for $(\Pi, \Delta)$ of size $\mathcal{O}(\Delta m + \abs{\Pi})$.
\end{theorem}

\begin{proof}
    We give the proof idea. We first transform the $\mpmsc$-program to a strongly equivalent $\mpmsc$-program, where proposition symbols and head predicates are fixed (Lemma \ref{lem:perform}). Then, we transform that program to a strongly equivalent $\mmsc$-program (Corollary \ref{lemma:mpmsc_to_mmsc}). From that program, we construct an $\mpc$ whose state string stores the truth values of head predicates and proposition symbols. 
    The circuit is essentially constructed directly from the inverse tree representations of clauses. Head predicates and proposition symbols in the scope of a diamond will correspond to input gates for bits sent by neighbouring nodes. Moreover, head predicates and propositions not in the scope of a diamond relate to input gates for the home node. In communication round zero, the circuit uses a subcircuit constructed from terminal clauses, and in later rounds, it uses a subcircuit constructed from iteration clauses, which can be achieved by applying Lemma \ref{lem:two_circuits}.

    Fix $\Delta \in \Z_+$ and a set of propositions $\Pi$. Let $\Lambda$ be a $\Pi$-program of $\mpmsc$ of size $m$. First, we transform $\Lambda$ to a strongly equivalent $\Pi$-program $\Lambda'$ of $\mpmsc$ of size $\ordo(\Delta m + \abs{\Pi})$, where propositions symbols and head predicates are fixed as in the statement Lemma \ref{lem:perform}. By Corollary \ref{lemma:mpmsc_to_mmsc}, for $\Lambda'$ there exists a strongly equivalent $\Pi$-program $\Gamma$ of $\mmsc$ of size $\ordo(\abs{\Lambda'}) = \ordo(\Delta m + \abs{\Pi})$ such that the modal depths of the terminal clauses of $\Gamma$ are zero, and the modal depths of iteration clauses are at most one. 

Let $\Pi' \subseteq \Pi$ be the set of propositions in $\Pi$ which appear in an iteration clause of the program, and let $k$ be the number of head predicates in $\Gamma$. Suppose that $X_1,\dots , X_k$ are the head predicates of $\Gamma$, given in the order $<^{\mathrm{VAR}}$. We assume that $p_1,\dots , p_{\abs{\Pi}}$ enumerate the
proposition symbols in $\Pi$ in the order $<^{\mathrm{PROP}}$. 
We will construct from $\Gamma$ a message passing circuit $C_{\Gamma}$ for $(\Pi, \Delta)$ with state length $\abs{\Pi'} + 1 + k$.
The last $k$ bits of a state string of $C_{\Gamma}$ will encode the truth values of $X^n_i$ in every round $n$. The first $\abs{\Pi'}$ bits will encode truth values of propositions---thereby being static or constant---during computation (these bits will be used for technical convenience to help when dealing with diamond operators). The one extra bit indicates whether we have already evaluated the terminal clauses of $\Gamma$ or not, and it also indicates whether a node has received a message from its neighbours or not. Having constant truth values, the propositions in $\Pi'$ are trivial to handle: 
for every input gate $G_{\text{in}}$ for a proposition symbol, we add a fresh $\land$-gate and connect $G_{\text{in}}$ to it in order to get an identity transformation.
The rest of the construction takes more work. Intuitively, the bit of a head predicate $X_i$ is updated by two subcircuits: in the first round by a subcircuit that corresponds to the terminal clause of $X_i$ and in subsequent rounds by a subcircuit that corresponds to the iteration clause.

To construct $C_{\Gamma}$, we begin by showing how to translate the terminal
clauses of $\Gamma$ to corresponding
circuits. These circuits will then be combined to form a circuit $C_0$ which will later on be attached to be part of the final circuit $C_{\Gamma}$. Now, let $X_i(0) \coloneq \varphi_i$ be a terminal clause of $\Gamma$. Note that thus $\varphi_i$ is free of diamonds. Let us construct a circuit $C_{(i,0)}$ for the formula $\varphi_i$. For each $j \in [\, \abs{\Pi}\, ]$, let $I^{\Pi}_{(j,0)}$ denote an input gate whose role will---intuitively---be to read the truth value of the symbol $p_j$ at the home node of the circuit. We transform $\varphi_i$ to its inverse tree representation. Then, for each $j$, we replace each instance of the proposition symbol $p_j$ in the tree with the input gate $I^{\Pi}_{(j,0)}$. Now, we have the circuit $C_{(i,0)}$ for $\varphi_i$ ready. The output gate is the gate with fan-out zero. We combine these circuits $C_{(i,0)}$ for all $i$ to a single circuit $C_0$ such that they share input gates, i.e., all repetitions of a gate $I^{\Pi}_{(j,0)}$ are combined to form a single input gate. The output nodes of $C_0$ are the output nodes of the circuits $C_{(i,0)}$ in the obvious order such that $C_{(j',0)}$ is before $C_{(j'',0)}$ if $j'< j''$.

We will next construct a circuit for each iteration clause of $\Gamma$. These will be then combined to a single circuit $C_1$, to be ultimately connected to be part of $C_{\Gamma}$. The inputs of a circuit corresponding to an iteration clause will correspond to the messages of neighbours accessed by diamonds $\Diamond_i$. Now, let $X_{k'} \coloneq \psi_{k'}$ be an iteration clause of $\Gamma$. Let us construct a circuit $C_{(k',1)}$ for the schema $\psi_{k'}$. Let $j \in [\Delta]$ and $i \in [k]$. Let $I_{(i,j)}$ denote an input gate of $C_{(k',1)}$ whose role will---intuitively---be to read the truth value that $X_i$ has at the $j$th neighbour. Similarly, by $I_{(i,0)}$ we denote the input gate reading $X_i$ at the home node. We let $I_{(0,j)}$ denote the input gate whose role will---intuitively---be to tell if the home node received message from the $j$th neighbour. Also, we let $I^{\Pi'}_{(i,j)}$ denote an input gate reading the value of the symbol $p_i\in\Pi'$ at the neighbour $j$, noting that ``neighbor $0$'' refers to the home node. Now, we first transform $\psi_{k'}$ to its inverse tree representation. Then, we replace propositions and head predicates in the tree with input gates as follows. First, if a head predicate $X_i$ (respectively, proposition symbol $p_i$) is a  descendant of some $\Diamond_j$-node in the tree, then we use $I_{(i,j)}$ (resp., $I^{\Pi'}_{(i,j)}$) as the replacing input gate. (Note here that the descendant does not have to be a child of the $\Diamond_j$-node.) If the head predicate (resp., proposition symbol) is not a descendant of a diamond node, we use $I_{(i,0)}$ (resp., $I^{\Pi'}_{(i,0)}$) as the replacing input gate. After all the head predicates and propositions have been replaced, we replace the $\Diamond_j$-nodes with a $\land$-gate, connecting the parent of each diamond node directly to its child, and we also connect $I_{(0,j)}$ to that replacing $\land$-gate. Now, if $I_{(0,j)}$ gets $0$ as input then the replacing $\land$-gate will also output $0$. Now, the circuit $C_{(k',1)}$ for $\psi_{k'}$ is ready. The output is the gate with fan-out zero. To construct $C_1$, we combine the circuits $C_{(k', 1)}$ for all $k'$ by making sure that for all $i$ and $j$, repetitions of input gates $I_{(i,j)}$ become a single input gate. Similarly, repetitions of $I_{(i,j)}^{\Pi}$ are removed by combining input gates. The output gates of $C_{(j,1)}$ are the output-gates of the circuits $C_{(k',1)}$ ordered so that the output of $C_{(k'',1)}$ is before that of $C_{(k''',1)}$ when $k''< k'''$.

The circuit $C_0$ now mimics the initializing computation round of the program and $C_1$ mimics the remaining rounds. Next, we will combine $C_0$ and $C_1$ to one circuit $C_{\Gamma}$, with $C_0$ computing the very first round and $C_1$ all the later rounds. This can be constructed by applying Lemma \ref{lem:two_circuits} to $C_0$ and $C_1$. The circuit $C_{\Gamma}$ uses the input gates $I^{\Pi}_{(i,0)}$, $I^{\Pi'}_{(i',j)}$, $I_{(i'',j)}$ and one input gate that is added by Lemma \ref{lem:two_circuits} (in total, this makes $\abs{\Pi} + (\abs{\Pi'} + 1 + k)(\Delta + 1)$ input gates). The order of these input gates is obvious. The attention and print bits are corresponding indexes of attention and print predicates of $\Gamma$. 
Now, $C_{\Gamma}$ receives a non-zero message from its neighbours since Lemma \ref{lem:two_circuits} adds a $1$ in front of the output of the circuit. This makes sure that if some $j$th neighbour does not exist, then $\land$-gates that replaced $\Diamond_j$ nodes output $0$ as wanted.  
Also, as we mentioned in the beginning of the proof, we have the bits for $\Pi'$ that should remain static in every evaluation round. The input gates for these bits were denoted by $I^{\Pi'}_{(i,0)}$. At this stage, we have not yet defined corresponding output gates for the bits for $\Pi'$, so we will manipulate $C_{\Gamma}$ as follows. For each input gate $I^{\Pi'}_{(i,0)}$, it is trivial to construct an identity circuit that gives the input to $I^{\Pi'}_{(i,0)}$ as an output. The output gates of these identity circuits are added in front of the output of $C_{\Gamma}$ in the obvious order.

The circuit $C_{\Gamma}$ has been constructed to simulate the program $\Gamma$, and a straightforward, but tedious, induction shows that the circuit and the program are strongly equivalent. As the circuit construction is based on simple modifications of inverse tree representations of clauses of the programs, it is easy to see that the size of the circuit is $\mathcal{O}(\Delta m + \abs{\Pi})$. 

\end{proof}

\section{Linking standard MSC to message passing circuits}

\subsection{Using a clock in a program}\label{clock}
    
Our next goal is to define a program for simulating a clock that has a \emph{minute hand} and a \emph{second hand}. The program for the minute hand works on binary strings of length $\ell$, while the second hand helps in updating the minute hand.
The minute hand starts with the string $0^{\ell}$ and goes through all $\ell$-bit strings in the canonical lexicographic order. To update the minute hand, we use the second hand as follows. We look for the first zero (from the right) of a string. We call this zero---or its position in the string---the \emph{flip point}. 
The bits to the left of the flip point are kept as they are, and the remaining bits are flipped, including the flip point. For example, $111011$ has the third bit from the right as a flip point, and the string gets updated to $111100$. After
reaching the string $1^{\ell}$,
the minute hand starts again from $0^{\ell}$.

We will use schema variables $M_1,\dots , M_{\ell}$ to encode bits of the minute hand such that $M_1$ records the rightmost bit, $M_2$ the second bit from the right, and so on. The second hand is defined using the head predicates $S_1 , \ldots, S_{\ell}$ and $S_{\text{changing}}$. 
Consider an example (see the figure) where $\ell=4$. The leftmost column gives the computation rounds, and then we have the bit strings for the minute hand $(M_4,M_3,M_2,M_1)$; in the middle we have the strings for the second hand $(S_4,S_3,S_2,S_1)$; the rightmost column gives the values of variable $S_{\text{changing}}$. 

{\footnotesize
\setlength{\arraycolsep}{2.5pt}
\[
\begin{array}{lllllllllllllll}
{}&\qquad\qquad & M_4&M_3&M_2&M_1\qquad\qquad & & & S_4&S_3&S_2&S_1\qquad\qquad & & &S_{\text{changing}}\\
0.& & 0&0&0&0 & & & 0&0&0&0 & & &0\\
\hdashline
1.& & 0&0&0&1 & & & 0&0&0&0 & & &1\\
2.& & 0&0&0&1 & & & 0&0&0&1 & & &1\\
\rowcolor{blue!5}
3.& & 0&0&0&1 & & & 0&0&0&1 & & &0\\
\hdashline
4.& & 0&0&1&0 & & & 0&0&0&0 & & &1\\
\rowcolor{blue!5}
5.& & 0&0&1&0 & & & 0&0&0&0 & & &0\\
\hdashline
6.& & 0&0&1&1 & & & 0&0&0&0 & & &1\\
7.& & 0&0&1&1 & & & 0&0&0&1 & & &1\\
8.& & 0&0&1&1 & & & 0&0&1&1 & & &1\\
\rowcolor{blue!5}
9.& & 0&0&1&1 & & & 0&0&1&1 & & &0\\
\hdashline
10.& & 0&1&0&0 & & & 0&0&0&0 & & &1\\
\textbf{\vdots}& & {}&{}&\textbf{\vdots}&{} & & & {}&{}&\textbf{\vdots}&{} & & &\textbf{\vdots}\\
\hdashline
k.& & 0&1&1&1 & & & 0&0&0&0 & & &1\\
k+1.& & 0&1&1&1 & & & 0&0&0&1 & & &1\\
k+2.& & 0&1&1&1 & & & 0&0&1&1 & & &1\\
k+3.& & 0&1&1&1 & & & 0&1&1&1 & & &1\\
\rowcolor{blue!5}
k+4.& & 0&1&1&1 & & & 0&1&1&1 & & &0\\
\hdashline
k+5.& & 1&0&0&0 & & & 0&0&0&0 & & &1\\
\textbf{\vdots}& & {}&{}&\textbf{\vdots}&{} & & & {}&{}&\textbf{\vdots}&{} & & &\textbf{\vdots}\\
\end{array} 
\]}

Now, row $0$ is a special case, as round $0$ is the round of initiation. After that, the computation proceeds in vertical blocks of rounds separated by dashed lines, with the block of rounds 1-3 being first, then the block for rounds 4-5, then 6-9, et cetera. The blocks have different heights. 
The strings for $(M_4,M_3,M_2,M_1)$ correspond to the minute hand. We observe that the minute hand is constant inside each block and gets increased in the standard lexicographic way when changing blocks. The variable $S_{\text{changing}}$ is $1$ within each block, with the exception of the last rows of blocks (highlighted in the figure). Thus, $S_{\text{changing}}$ indicates that we should start a new block. The string for $(S_4,S_3,S_2,S_1)$ is always of type $0^i 1^j$, and the number of bits $1$ increases until we reach the penultimate row. This reflects the idea that $(S_4,S_3,S_2,S_1)$ is copying the string for $(M_4,M_3,M_2,M_1)$ from right to left, until we reach the flip point, i.e., the first $0$ (from the right) in the string for $(M_4,M_3,M_2,M_1)$.

Now, it is easy to generalize this from the case $\ell = 4$ to the general case. 
We also need a clock that is always one step ahead of the basic clock. We define this \emph{forward clock} as follows. For each head predicate $M_i, S_i, S_{\text{changing}}$ of the basic clock, we define a fresh symbol $M_i', S_i', S_{\text{changing}}'$. The forward clock uses only these fresh symbols. The iteration clauses are obtained by copying the corresponding rules and changing the predicates to fresh ones. The terminal clauses of the forward clock are obtained similarly by copying the terminal clauses of the basic clock, with the exception of the clauses for $M'_1$ and $S'_{\text{changing}}$, which we set to $M'_1 (0) \coloneq \top$ and $S'_{\text{changing}} (0) \coloneq \top$. This means the forward clock is otherwise the same as the basic one, but it starts from the string $0^{\ell -1}1$ instead of $0^{\ell}$. Note that the forward clock starts with $S_{\text{changing}}$ as true since otherwise in round $1$ the minute hand of the forward clock would correspond to the string $0^{\ell}$. It is similarly possible to define a \emph{double forward clock} beginning with $0^{\ell-2}10$. 

Let us fix the program formally. Let $\Pi$ be a finite set of propositions and $\ell \coloneqq \abs{\Pi_1}$.
For every $i \in [\ell]$, we let $S_i (0) \coloneq \bot$. Also, we let $S_1 \coloneq_{S_{\text{changing}}} M_1; \bot$ and $S_{i} \coloneq_{S_{\text{changing}}} S_{i-1} \land M_{i}; \bot$ 
for $i \in \{2,\dots , \ell \}$. 
To define the rules for $S_{\text{changing}}$, notice that the point of $S_{\text{changing}}$ is to become false precisely on those rounds where the value for $(S_{\ell},\dots , S_1)$ becomes repeated, i.e., identical to the corresponding string from the previous round. Since the head predicates $S_{i}$ (for $i\not=1$) are updated according to the rule body $S_{i - 1} \land M_i$ (or $M_1$ for $i=1$), we can ensure that the current and previous string for
$(S_{\ell},\dots , S_1)$ are different by making sure that either
$S_{i}$ and $S_{i - 1} \land M_i$ fail to be
equivalent for some $i\not = 1$ or, alternatively, 
$M_1$ and $S_1$ fail to be equivalent. 
This can be forced by the iteration rule
\[S_{\text{changing}} \coloneq_{S_{\text{changing}}}  \neg (S_{1} \leftrightarrow M_1) \lor \bigvee_{1 < i \le \ell}^{} \neg \left( S_{i} \leftrightarrow \left( S_{i - 1} \land M_{i} \right)  \right) ; \top. \]
The corresponding terminal rule is $S_{\text{changing}}(0) \coloneq \bot$.

We then fix the minute hand with the head predicates $M_1,  \ldots, M_{\ell}$. For all $i \in [\ell]$, we let $M_i (0) \coloneq \bot$. We also let $M_1 \coloneq_{S_{\text{changing}}} M_1; \neg M_1$, and for $i \in \{2,\dots , \ell\}$, we define $M_{i} \coloneq_{S_{\text{changing}}} M_{i}; \psi_{\text{change } i}$, where 
\[\psi_{\text{change }i} \coloneqq \left( S_{i} \land \neg M_{i}\right) \lor \left( S_{i-1} \land \neg S_{i} \land \neg M_{i} \right) \lor \left( \neg S_{i-1} \land \neg S_{i} \land M_{i} \right).\]

The first disjunct of $\psi_{\text{change }i}$ takes care of the values to the right of the flip point; the second disjunct changes the flip point to one; and the last disjunct keeps the bits as they are to the left of the flip point.

We have now defined a clock for circulating the bit strings for $(M_{\ell},\dots , M_1)$. The clock is of linear size in relation to the predicates $M_i$. \\

\subsection{Simulating multimodal diamonds}\label{Ordering neighbors by propositions}

To simulate $\mpmsc$ (or $\mmsc$) in $\msc$, we will need to simulate each $\Diamond_i$ with $\Diamond$ only. For this, we will use $\ID$s and clocks. The idea is to scan through the neighbours one by one in the order given by the $\ID$s. To keep our translations linear in size, different diamonds $\Diamond_i$ will be ``read'' in different rounds. We note that we could speed up reading the diamonds, especially for constant out-degree models, but the approach below suffices for the current study.
We begin with an example that conveys the intuition of the simulation. In the example, we assume that $\Pi_1 = \{p_1,p_2,p_3\}$ and $\ell = \abs{\Pi_1} = 3$ and consider a program where the maximum subindex of a diamond is $I = 3$. 
Let us examine a node with three neighbours with the identifiers $000$, $010$ and $111$. Consider the following array.

{\footnotesize
\setlength{\arraycolsep}{2.5pt}
\[
\begin{array}{llllllllllllllllllllll}
{}&\qquad &{} &X_{\ID} &{}\qquad\qquad & & & M_3&M_2&M_1\qquad\qquad & & M'_3&M'_2&M'_1 & & &X_{\text{reset}}\qquad & &X_{\text{not same}} &N_1 &N_2 &N_3\\
0.& & \neg p_3&\neg p_2&\neg p_1 & & & 0&0&0 & & 0&0&1 & & &0 & &1 &0 &0 &0\\
\hdashline
1.& & \neg p_3&\neg p_2&\text{ }\ p_1 & & & 0&0&1 & & 0&0&1 & & &0 & &1 &1 &0 &0\\
2.& & \neg p_3&\neg p_2&\text{ }\ p_1 & & & 0&0&1 & & 0&0&1 & & &0 & &0 &1 &0 &0\\
3.& & \neg p_3&\neg p_2&\text{ }\ p_1 & & & 0&0&1 & & 0&1&0 & & &0 & &0 &1 &0 &0\\
\hdashline
4.& & \neg p_3&\text{ }\ p_2&\neg p_1 & & & 0&1&0 & & 0&1&0 & & &0 & &1 &1 &0 &0\\
5.& & \neg p_3&\text{ }\ p_2&\neg p_1 & & & 0&1&0 & & 0&1&1 & & &0 & &0 &1 &1 &0\\
\hdashline
6.& & \neg p_3&\text{ }\ p_2&\text{ }\ p_1 & & & 0&1&1 & & 0&1&1 & & &0 & &1 &1 &1 &0\\
{}& & {}&\textbf{\vdots}&{} & & & {}&\textbf{\vdots}&{} & & {}&\textbf{\vdots}&{} & & &\textbf{\vdots}& &\textbf{\vdots}&{} &\textbf{\vdots}&{}\\
k.& & \text{}\ p_3&\text{ }\ p_2&\neg p_1 & & & 1&1&0 & & 1&1&1 & & &0 & &0 &1 &1 &0\\
\hdashline
k+1.& & \text{}\ p_3&\text{ }\ p_2&\text{ }\ p_1 & & & 1&1&1 & & 1&1&1 & & &0 & &1 &1 &1 &0\\
k+2.& & \text{}\ p_3&\text{ }\ p_2&\text{ }\ p_1 & & & 1&1&1 & & 1&1&1 & & &0 & &0 &1 &1 &1\\
k+3.& & \text{}\ p_3&\text{ }\ p_2&\text{ }\ p_1 & & & 1&1&1 & & 1&1&1 & & &0 & &0 &1 &1 &1\\
k+4.& & \text{}\ p_3&\text{ }\ p_2&\text{ }\ p_1 & & & 1&1&1 & & 1&1&1 & & &0 & &0 &1 &1 &1\\
k+5.& & \text{}\ p_3&\text{ }\ p_2&\text{ }\ p_1 & & & 1&1&1 & & 0&0&0 & & &1 & &0 &1 &1 &1\\
\hdashline
k+6.& & \neg p_3&\neg p_2&\neg p_1 & & & 0&0&0 & & 0&0&0 & & &0 & &1 &0 &0 &0\\
k+7.& & \neg p_3&\neg p_2&\neg p_1 & & & 0&0&0 & & 0&0&1 & & &0 & &0 &1 &0 &0\\
\hdashline
k+8.& & \neg p_3&\neg p_2&\text{ }\ p_1 & & & 0&0&1 & & 0&0&1 & & &0 & &1 &1 &0 &0\\
{}& & {}&\textbf{\vdots}&{} & & & {}&\textbf{\vdots}&{} & & {}&\textbf{\vdots}&{} & & &\textbf{\vdots}& &\textbf{\vdots}&{} &\textbf{\vdots}&{}\\
\end{array}
\]}

The dashed lines define blocks for the clock $(M_3,M_2,M_1)$ so that again the corresponding bit strings increase from block to block. Only the minute hand predicates are shown. The tuple $(M_3', M_2', M_1')$ encodes a forward clock running one step ahead. Intuitively, $X_{\text{reset}}$ is a flag that shows when to reset the flags $N_i$ (to be explained later) back to zero. Formally, $X_{\text{reset}}$ is true precisely when $(M_3,M_2,M_1) = (1,1,1)$ and $(M_3', M_2', M_1') = (0,0,0)$, which happens in the round when the basic clock corresponds to the
string $111$ for the last time. 
The flag $X_{\text{not same}}$ is true
precisely \emph{after} each round where the basic and forward clock differ in at least one bit, that is, the strings for $(M_3,M_2,M_1)$ and $(M_3', M_2', M_1')$ are different. ($X_{\text{not same}}=1$ in round $0$ for technical convenience.) The purpose of $X_{\text{not same}}$ is to help us deal with the predicates $N_i$ in the following way. First, we have precisely the three predicates $N_1,N_2,N_3$ because the maximum subindex of a diamond is $I = 3$. The predicate $N_i$ becomes true in rounds $n$ such that the string $(M_3,M_2,M_1)$ has precisely matched the ID of the $i$th neighbouring node in round $n-1$. For example, our $2$nd neighbour has ID $010$, and $(M_3,M_2,M_1) = (0,1,0)$ in round $4$, so $N_2$ becomes true in round $5$. The predicates $N_i$ stay true until resetting in the round directly after the round $X_{\text{reset}}$ is true. The predicate $X_{\mathrm{ID}}$ is true precisely in those rounds where $(M_3,M_2,M_1)$ encodes the $\mathrm{ID}$ of the current node, that is, we have $(M_3,M_2,M_1) = (b_3,b_2,b_1)$ if that $\mathrm{ID}$ is $b_3b_2b_1\in \{0,1\}^*$, meaning that the truth values of $p_3,p_2,p_1$ are $b_3,b_2,b_1$. The first three columns of the figure encode the vectors $$(\varphi_3, \varphi_2, \varphi_1)\in \{p_3, \neg p_3\}\times \{p_2, \neg p_2\}\times \{p_1, \neg p_1\}$$ with the property that in the round $n$, the variable $X_{\ID}$ is equivalent to the conjunction $\varphi_3 \wedge \varphi_2 \wedge \varphi_1$. For example, in round six, $(M_3, M_2, M_1) = (0,1,1)$ and $(\varphi_3, \varphi_2, \varphi_1) = (\neg p_3, p_2, p_1)$. A key intuition in simulating a diamond $\Diamond_i$ with $\Diamond$ relates to scanning the $\mathrm{ID}$s of \emph{neighbouring} nodes via the dynamically changing truth value of $X_{\mathrm{ID}}$ at the \emph{neighbouring} nodes.

We specify the program formally in the general case for a fixed $\Pi_1$ and $I \in \Z_+$. Let $\ell = \abs{\Pi_1}$. We construct a program for the head predicates $(M_{\ell}, \ldots, M_1)$ and $(N_I, \ldots, N_1)$. Assume that $p_{1}, \ldots, p_{\ell}$ enumerate the propositions in $\Pi_1$ in the order $<^{\mathrm{PROP}}$. We specify a
clock via the tuple $(M_{\ell},\dots , M_1)$ of 
head predicates and a forward
clock via $(M_{\ell}', \dots , M_1')$. For technical convenience, we even define a double forward clock via $(M_{\ell}'', \dots , M_1'')$. 
The program for the head predicate $X_{\text{reset}}$ is
$
X_{\text{reset}}(0) \coloneq \bot,\ \ X_{\text{reset}} \coloneq \bigwedge_{i \leq \ell} \neg M''_i \land \bigwedge_{i \leq \ell} M'_i.
$
The program for the head predicate $X_{\text{not same}}$ is
$
X_{\text{not same}}(0) \coloneq \top$,\ \ $ X_{\text{not same}} \coloneq \bigvee_{i \leq \ell} \neg( M_i \leftrightarrow M'_i).
$
The program for $X_{\ID}$ is $X_{\ID} (0) \coloneq \bigwedge_{i \leq \ell} \neg p_i$, $X_{\ID} \coloneq \bigwedge_{i \leq \ell} (M'_i \leftrightarrow p_i)$.

For every $i \in [I]$, the rules for $N_i$ are 
$
N_i (0) \coloneq \bot$ 
and \[N_{i}\coloneq_{X_{\text{reset}},\ \neg N_{i} \land N_{i-1} \land X_{\text{not same}}} \bot; \Diamond X_{\ID}; N_{i}\]
where $N_{i-1}$ is simply deleted when $i - 1 = 0$.
The first condition $X_{\text{reset}}$ is obvious. The second condition $\neg N_i \land N_{i-1} \land X_{\text{not same}}$ checks that we have already found the neighbour $i-1$, i.e., $N_{i-1}$ is true (note that the variable $N_0$ does not exist, so no such check is needed for $N_1$). Moreover, this condition checks that we have not yet found neighbour $i$, that is, $N_i$ is false and also that the clock has been updated in the previous round, i.e., $X_{\text{not same}}$ is true. 
This makes sure that within one block, precisely one flag $N_i$ can become true.

\textbf{The diamond $\Diamond_{i}$ macro} for $\varphi$ and with respect to $(N_i,N_{i-1}, X_{\mathrm{ID}})$
 is then the schema 
$\Diamond_{i}\varphi \coloneqq (\neg N_{i} \land N_{i-1} \land \Diamond \left( \varphi \land X_{\ID} \right) ).$
Note that in the $\Diamond_1$ macro, $N_0$ is omitted
(meaning $N_{i-1}$ is deleted in the above). Informally, the $\Diamond_i$ macro is ``reliable'' only at specific times of computation and thus must be used with care. This is demonstrated in the proof of the following lemma, given in the appendix. Also, the subsequent result (Theorem \ref{thr:mpmsc_to_msc}) then follows immediately by Lemma \ref{lemma:cmsc_to_msc}.

\begin{lemma}\label{ridofsubindiceslemma}
    Given $\Pi$ and a $\Pi$-program of $\mpmsc$ of size $m$ where the maximum subindex of a diamond is $I$, we can construct an equivalent $\Pi$-program of $\cmsc$ of size $\ordo(I + \abs{\Pi_1} + m)$. The computation time is $\ordo(2^{\abs{\Pi_1}})$ times the computation time of the $\mpmsc$-program. 
\end{lemma}


\begin{theorem}\label{thr:mpmsc_to_msc} 
Given $\Pi$ and a $\Pi$-program of $\mpmsc$ of size $m$, where the maximum subindex in a diamond is $I$, we can construct an equivalent $\Pi$-program of $\msc$ of size $\ordo(I + \abs{\Pi_1} + m)$. The computation time is $2^{\ordo(\abs{\Pi_1})}$ times the computation time of the $\mpmsc$-program.
\end{theorem}

\subsection{A normal form for MSC}

A program of \textbf{$\msc[1]$} is a program of $\msc$ such that the modal depth of terminal (respectively, iteration) clauses is zero (resp., at most one).
We begin with the following lemma. The proof, given in the appendix, is based on (1) making terminal clauses part of iteration clauses by suitably using conditional $\msc$ and (2)
translating $\cmsc$ to $\msc$.

\begin{lemma}\label{lem:term_zero} 
    For every $\Pi$-program $\Lambda$ of $\msc$, there exists an equivalent $\Pi$-program of $\msc$ where the modal depth of terminal clauses is zero. The size of the program is linear in $|\Lambda|$. The computation time is linear in the computation time of $\Lambda$.
%
%
%
\end{lemma}

We then show that the modal depth of iteration clauses can be reduced to one. The proof, given in the appendix, is based on the simple strategy of using fresh head predicates $X_{\Diamond\psi}$ for each subschema of type $\Diamond\psi$ of the original program. With that strategy and also by using Lemma \ref{lem:term_zero}, we can prove the following normal form result.

\begin{theorem}\label{thr:msc1}
    For every $\Pi$-program $\Lambda$ of $\msc$, there exists an equivalent $\Pi$-program of $\msc[1]$. The size of the $\msc[1]$-program is $\ordo(\abs{\Lambda})$ and the computation time of the program is $\ordo(max(1, \md(\Lambda)))$ times the computation time of $\Lambda$.
\end{theorem}


%
%
%

\subsection{Equivalence and time loss}

We are now ready to link $\msc$ to $\mpc$s.
In Section \ref{sec:mpc_to_mpmsc}, we proved Theorem \ref{thr:mpc_to_mpmsc} that shows we can translate $\mpc$s to communication equivalent $\mpmsc$-programs of size linear in the size of the $\mpc$. On the other hand, Theorem \ref{thr:mpmsc_to_msc} shows that we can translate any $\mpmsc$-program to an equivalent program of $\msc$. Thus, we get the following.
\begin{theorem}\label{thr:mpc_to_msc}
    Given an $\mpc$ for $(\Pi, \Delta)$, we can construct an equivalent $\Pi$-program of $\msc$. For a constant bound $c$ for the fan-in of $\mpc$s, the size of the program is linear in the size of the circuit. The computation time is $\ordo(d + 2^{\abs{\Pi_1}})$ times the computation time of the $\mpc$, where $d$ is the depth of the $\mpc$.
\end{theorem}

 Theorem \ref{thr:mpmsc_to_mpc} showed that we can translate an $\mpmsc$-program to a strongly equivalent $\mpc$. 
Theorem \ref{thr:msc1} showed how to translate an $\msc$-program to a strongly equivalent $\msc[1]$-program, implying that translating an $\msc$-program to an $\mpmsc$-program can be done without blowing up program size too much. These results imply the following corollary. 
\begin{theorem}\label{thr:msc_to_mpc}
Given $\Pi$, $\Delta$ and a $\Pi$-program of $\msc$ of size $m$, 
there exists an equivalent $\mpc$ for $(\Pi, \Delta)$ of size $\ordo(\Delta m + \abs{\Pi})$. The computation time is $\ordo(\max(1,d))$ times the
computation time of the $\msc$-program, where $d$ is the modal depth of the $\msc$-program.  
\end{theorem}

%
%
%

%
%
%
%
%

By the above results, we observe that
problems in $\mathrm{DCC}_{\Delta}^c[\log n]$ can be alternatively described with sequences of $\msc$-programs. 

Finally, we note that also Theorem \ref{thr:mpc_to_mpmsc} is one of our main results. It reminds us that
communication time is indeed a different concept from computation time.

We note that in addition to standard distributed computing, we can use $\msc$ for modeling, .e.g., neural networks. Real numbers can be naturally modeled via using head predicates to encode floating-point numbers. Identifiers (and some lighter solutions) can be used for keeping track where messages were sent from, and biases are naturally encoded into proposition symbols. Activation functions can be approximated, and of course standard arithmetic calculations as well. We leave this for future work.

\section{Coloring based on Cole \& Vishkin}\label{C&V}

\subsection{Summary of the Cole-Vishkin algorithm}\label{Summary of the Cole-Vishkin algorithm}

A \textbf{graph} is an ordered pair $G = (N, E)$, where $N$ is a set of \textbf{nodes}, and $E$ is a set of \textbf{edges} $(v, u) \in N \times N$. The edge relation is symmetric (if $(v, u) \in E$ then $(u, v) \in E$) and irreflexive ($(v, v) \notin E$ for all $v \in N$). We call $v$ and $u$ \textbf{neighbors} if $(v, u) \in E$.

An \textbf{orientation} of graph $G$ is an ordered pair $(N, E')$, where for each pair $(v, u) \in E$, either $(v, u) \in E'$ or $(u, v) \in E'$, but not both. If $(v, u) \in E'$, then we call $u$ the \textbf{parent} of $v$ and $v$ the \textbf{child} of $u$. We call the tuple $(v_1, \dots, v_k)$, $k > 1$, a \textbf{directed cycle} if $v_1 = v_k$ and $(v_1, v_2), \dots, (v_{k - 1}, v_k) \in E'$.

A \textbf{forest decomposition} of graph $G$ is an ordered set $\mathcal{F}$ of ordered pairs $F_i = (N, E_i)$ such that there is an orientation $(N, E')$ of $G$ with no directed cycles, where $\bigcup E_i = E'$, and $E_i \cap E_j = \emptyset$ for all $i \neq j$. We call each $F_i \in \mathcal{F}$ an \textbf{oriented forest} (of $G$).

We denote $\log(k) \coloneqq \lfloor \log_2(k) \rfloor + 1$, so that $\log(k)$ is the number of bits of $k$ in binary. Let us consider a graph $G$, where the number of nodes is $n$, and the maximum degree of a node is $\Delta$. We describe a color reduction algorithm derived from \cite{barenboim2013distributed}, pages 36-37, where each node begins with a unique identifier ($\ID$) from the set $\{1, \dots, n\}$; this is its first color. The algorithm produces a $(\Delta + 1)$-coloring in $3^\Delta + \log^*(n) + \mathcal{O}(1)$ communication rounds.

Using the identifiers, each node starts by orienting its adjacent edges so that each edge is pointing at the node with the higher $\ID$; this gives us an orientation of graph $G$. Additionally, each node labels all of its outgoing edges according to the $\ID$s of the connected nodes in ascending order with the labels $1, \dots, \Delta$. By sorting the edges into separate sets according to their labels, we get a forest decomposition of $\Delta$ oriented forests: $\mathcal{F} = (F_1, \dots, F_{\Delta})$.

In each oriented forest, each node runs the Cole-Vishkin algorithm \cite{cole1986deterministic, goldberg1987parallel} (or CV algorithm). The node first compares its color to that of its parent. Colors are coded into binary in the ordinary way. It finds the rightmost bit that is different from that of its parent; let us assume it is the $i$th bit. If the $i$th bit of the node is $b \in \{0, 1\}$, the first bit of its new color is also $b$. The remaining bits of the new color are the bits of the binary representation of the number $i$. A root does not have a parent and instead compares its color to the color $0$.

By repeating this algorithm, each node is able to reduce the number indexing its current color in each oriented forest. If the highest color of a node in a forest has length $k$, the next color will have a length at most $\log(k) + 1$ after the next communication round. After $\log^*(n) + \mathcal{O}(1)$ communication rounds, each node will have a color of length at most $3$ in each forest. This is a $7$-coloring in each forest because the color $0$ is not used.

After this, each $7$-coloring can be further reduced to a $3$-coloring with a technique called ``shift-down'' in just $\mathcal{O}(1)$ communication rounds. These colorings can be concatenated into a $3^{\Delta}$-coloring of the whole graph, which can then be reduced to a $(\Delta + 1)$-coloring with a technique called ``basic color reduction'' in less than $3^\Delta$ communication rounds.

In the next section, we outline an $\mpmsc$-program that simulates the beginning of this coloring algorithm (full details in the appendix). When defining this program, the sets of print predicates, attention predicates and, thus, appointed predicates are all the same. Our notation will differ slightly from before. Previously, head predicates were denoted with single italic letters with indices (e.g., $X_i$). In the program that follows, they are not written in italic and may include multiple symbols (e.g., $\mathrm{HP1}_i$). For convenience, we may also use the same rule for formulae; whether a denotation refers to a variable or formula is made clear when it is introduced.

\begin{definition}\label{define coloring}
    Let $\Lambda$ be a ($\Pi, \Delta$)-program of $\mathrm{MPMSC}$ where for each Kripke-model $(W, R, V) \in \cK(\Pi, \Delta)$, each node $w \in W$ outputs $\overline{p}_{w}$ in some round $n \in \N$. We say that $\Lambda$ \textbf{defines a $k$-coloring} if for each node $w \in W$ in each Kripke-model $(W, R, V) \in \cK(\Pi, \Delta)$ there exist some $a \in [k - 1]_0$ and $b \in \N$ such that $\overline{p}_{w} = 0^{a}10^{b}$. This coloring is \textbf{proper} if $\overline{p}_{w} \neq \overline{p}_{v}$ for each $(w, v) \in R$ in each Kripke-model $(W, R, V) \in \cK(\Pi, \Delta)$.
\end{definition}

\subsection{Simulating the CV algorithm with a program}\label{Simulating the C&V algorithm with a program}

Let $G = (N, E)$ be a graph with unique identifiers, where the number of nodes is $n$, and the maximum degree of a node is $\Delta$. Its \textbf{corresponding} Kripke model with identifiers over $\Pi = \{p_1, \dots, p_{\log(n)}\}$ is the model $(W, R, V)$, where $W = N$, $R = E$, and $v \in V(p_i)$ if and only if the $i$th bit of the $\mathrm{ID}$ of $v$ is $1$. Note that $(W, R, V) \in \cK(\Pi, \Delta)$.

The program starts by simulating the separation of the graph into oriented forests. For this, each node needs to know which of its neighbors have a higher $\ID$ than itself. This is done by comparing the bits of the $\ID$s one by one, which takes $\ell = \log(n)$ iterations. We count these iterations with variables $\mathrm{T}_1, \dots, \mathrm{T}_{\ell + 2}$. Each variable starts off as false. Variable $\mathrm{T}_1$ becomes permanently true after the first iteration, $\mathrm{T}_2$ after the second, etc.

To compare $\mathrm{ID}$s with neighbors, we define variables $\mathrm{I}_1, \dots, \mathrm{I}_\ell$ for the bits of a node's $\ID$ and variables $\mathrm{I}^\delta_1, \dots, \mathrm{I}^\delta_\ell$ ($1 \leq \delta \leq \Delta$) for the bits of its $\delta$th neighbor's $\ID$. In the terminal clauses, each $\mathrm{I}_i$ receives the truth value of the proposition $p_i$. In the first iteration, variables $\mathrm{I}_i$ stay the same while variables $\mathrm{I}^\delta_i$ receive the truth values of formulae $\Diamond_\delta \mathrm{I}_i$. For the next $\ell$ iterations, the variables rotate truth values, i.e., $\mathrm{I}^{(\delta)}_i$ receives the truth value of $\mathrm{I}^{(\delta)}_{i-1}$.

After one iteration round, we have the leftmost bit of a node's ($\delta$th neighbor's) $\ID$ in variable $\mathrm{I}^{(\delta)}_{\ell}$. After another iteration, this variable holds the second bit, etc. We determine if a node has a higher $\ID$ than its neighbor by finding the first bit that separates the $\ID$s; the node whose bit is $1$ has the higher $\ID$. We define variables $\mathrm{DIF}^{\delta}$ ($1 \leq \delta \leq \Delta$) for comparing the bits in the variables $\mathrm{I}_{\ell}$ and $\mathrm{I}^{\delta}_{\ell}$; if the truth values (i.e., bits) are different, $\mathrm{DIF}^{\delta}$ is true.

Finally, we define variables $\mathrm{HIGH}^{\delta}$ and $\mathrm{LOW}^{\delta}$ ($1\leq \delta \leq \Delta$). After $\mathrm{DIF}^{\delta}$ becomes true, $\mathrm{HIGH}^{\delta}$ turns true if a node has a higher $\ID$ than its $\delta$th neighbor; otherwise $\mathrm{LOW}^{\delta}$ turns true. After $\ell + 2$ iterations, either $\mathrm{HIGH}^{\delta}$ or $\mathrm{LOW}^{\delta}$ is true for every $\delta$ in every node.

It is now possible to define a formula $\langle \delta \rangle \varphi$ of length $\mathcal{O}(\Delta)$ that is true in node $v$ if and only if $\varphi$ is true in the parent of node $v$ in the oriented forest $F_{\delta}$. The formula guesses the index $i$ (if it exists) such that $\Diamond_i \varphi$ is true and $i$ is the $\delta$th index for which $\mathrm{LOW}^{i}$ is true.

Next, we extend the program to simulate the CV algorithm in each oriented forest. We compare the colors of nodes bit by bit and place communication rounds between the color comparisons. To do this, we use the clock from Section \ref{clock} with its associated variables. The strings of the clock have $\log(\ell) = \log\log(n)$ bits. A unique moment at the end of the clock's cycle is used to trigger the global communication rounds in this phase of the program.

To compare colors with neighbors, we define variables $\mathrm{B}^{\delta}_i$ ($1 \leq \delta \leq \Delta$) for the bits of a node's color in forest ${F_{\delta}}$ and variables $\mathrm{P}^{\delta}_i$ ($1 \leq \delta \leq \Delta$) for the bits of its parent's color, starting with the $\mathrm{ID}$s. Each variable $\mathrm{B}^{\delta}_i$ and $\mathrm{P}^{\delta}_i$ cycles through the bits of the color in forest $F_{\delta}$, i.e., $\mathrm{B}^{\delta}_i$ receives the truth value of $\mathrm{B}^{\delta}_{i+1}$, and $\mathrm{P}^{\delta}_i$ receives the truth value of $\mathrm{P}^{\delta}_{i+1}$. The variables rotate once every ``minute''. The minute hand always indicates the original position of the bits currently in $\mathrm{B}^{\delta}_1$ and $\mathrm{P}^{\delta}_1$ in binary. During the global communication round, the variables $\mathrm{B}^{\delta}_i$ update to the next color $\mathrm{N}^{\delta}_i$ (see below) and the variables $\mathrm{P}^{\delta}_i$ update to the next color of the parent with the formula $\langle \delta \rangle \mathrm{N}^{\delta}_i$. These formulae are always untrue for root nodes, which means that each root compares its color to the color $0$ just as described in section \ref{Summary of the Cole-Vishkin algorithm}.

We compare the bits of colors between child and parent nodes in variables $\mathrm{B}^{\delta}_{1}$ and $\mathrm{P}^{\delta}_{1}$ using the same variables $\mathrm{DIF}^{\delta}$ as before. Finally, we define variables $\mathrm{N}^{\delta}_i$ ($1 \leq \delta \leq \Delta$) that compute the next color of a node during each round of the clock. When $\mathrm{DIF}^{\delta}$ becomes true, the distinct bit of the node is copied from $\mathrm{B}^\delta_1$ to $\mathrm{N}^{\delta}_1$ and the bits of its position are copied from the minute hand variables $\mathrm{M}_i$ to the variables $\mathrm{N}^{\delta}_{i + 1}$. These values are preserved until the next global communication round when they are used to update variables $\mathrm{B}^{\delta}_i$ and $\mathrm{P}^{\delta}_i$.

The CV algorithm takes $L = \log^*(n) + 3$ communication rounds. We count them with an hour hand consisting of variables $\mathrm{H}_1, \dots, \mathrm{H}_L$; after $\log^*(n) + 2$ ``hours'' (i.e., cycles of the clock), $\mathrm{H}_L$ becomes true, and one ``hour'' later, the iteration of all variables is stopped.

For Lemma \ref{Cole-Vishkin lemma}, we define $7^\Delta$ ``appointed'' predicates corresponding to the color of each node when its colors from the $\Delta$ forests are concatenated by placing all the bits in a row. The following Lemma is now true, and extending the program with the shift-down and basic color reduction techniques, so is the Theorem thereafter (proofs are given in the appendix):
\begin{lemma}\label{Cole-Vishkin lemma}
There exists a formula of $\mpmsc$ with the following properties.
\begin{enumerate}
    \item It defines a proper $7^\Delta$-coloring.
    \item Ignoring appointed predicates, it has $\mathcal{O}(\Delta \log(n))$ heads and a size of $\mathcal{O}(\Delta^2 \log(n))$.
    \item The number of iterations needed is $\mathcal{O}(\log(n) \log\log(n) \log^*(n))$.
    \item There are exactly $\log^*(n) + 4$ global communication rounds.
\end{enumerate}
\end{lemma}

\begin{theorem}\label{Cole-Vishkin theorem}
There exists a formula of $\mpmsc$ 
with the following properties.
\begin{enumerate}
    \item It defines a proper $(\Delta + 1)$-coloring.
    \item It has $\mathcal{O}(\Delta \log(n)) + \mathcal{O}(3^\Delta)$ heads and its size is $\mathcal{O}(\Delta^2 \log(n)) + \mathcal{O}(3^\Delta)$.
    \item The number of iterations needed is $\mathcal{O}(\log(n) \log\log(n) \log^*(n)) + \mathcal{O}(\Delta 3^{\Delta})$.
    \item There are exactly $\log^*(n) + 3^\Delta - \Delta + 11$ global communication rounds.
\end{enumerate}
\end{theorem}

The global communication rounds of the program match the communication rounds of the algorithm, and if $\Delta$ is a constant, the size of the program is $\mathcal{O}(\log(n))$.

\medskip

\medskip 

\noindent
\textbf{Acknowledgments.} Veeti Ahvonen was supported by the Vilho, Yrjö and Kalle Väisälä Foundation of the Finnish Academy of
Science and Letters.  Antti Kuusisto also supported by the Academy of Finland project Theory of computational logics, grant numbers $324435$, $328987$, $352419$, $352420$, $352419$, $353027$. Antti Kuusisto was also supported
by the Academy of Finland project Explaining AI via Logic (XAILOG),
grant number $345612$. 


\bibliographystyle{plain}
\bibliography{references}

\begin{thebibliography}{10}

\bibitem{barcelo}
Pablo Barcel{\'{o}}, Egor~V. Kostylev, Mika{\"{e}}l Monet, Jorge P{\'{e}}rez,
  Juan~L. Reutter, and Juan~Pablo Silva.
\newblock The logical expressiveness of graph neural networks.
\newblock In {\em 8th International Conference on Learning Representations,
  {ICLR} 2020, Addis Ababa, Ethiopia, April 26-30, 2020}. OpenReview.net, 2020.

\bibitem{barenboim2013distributed}
Leonid Barenboim and Michael Elkin.
\newblock Distributed graph coloring.
\newblock {\em Synthesis Lectures on Distributed Computing Theory}, 11, 2013.

\bibitem{identifiers}
Benedikt Bollig, Patricia Bouyer, and Fabian Reiter.
\newblock Identifiers in registers - describing network algorithms with logic.
\newblock {\em CoRR}, abs/1811.08197, 2018.

\bibitem{cole1986deterministic}
Richard Cole and Uzi Vishkin.
\newblock Deterministic coin tossing with applications to optimal parallel list
  ranking.
\newblock {\em Information and Control}, 70(1):32--53, 1986.

\bibitem{goldberg1987parallel}
Andrew Goldberg, Serge Plotkin, and Gregory Shannon.
\newblock Parallel symmetry-breaking in sparse graphs.
\newblock In {\em Proceedings of the nineteenth annual ACM symposium on Theory
  of computing}, pages 315--324, 1987.

\bibitem{grohe}
Martin Grohe.
\newblock The logic of graph neural networks.
\newblock In {\em 36th Annual {ACM/IEEE} Symposium on Logic in Computer
  Science, {LICS} 2021, Rome, Italy, June 29 - July 2, 2021}, pages 1--17.
  {IEEE}, 2021.

\bibitem{weak_models}
Lauri Hella, Matti J{\"{a}}rvisalo, Antti Kuusisto, Juhana Laurinharju, Tuomo
  Lempi{\"{a}}inen, Kerkko Luosto, Jukka Suomela, and Jonni Virtema.
\newblock Weak models of distributed computing, with connections to modal
  logic.
\newblock {\em Distributed Comput.}, 28(1):31--53, 2015.

\bibitem{Kuusisto13}
Antti Kuusisto.
\newblock {Modal Logic and Distributed Message Passing Automata}.
\newblock In {\em Computer Science Logic 2013 (CSL 2013)}, volume~23 of {\em
  Leibniz International Proceedings in Informatics (LIPIcs)}, pages 452--468,
  2013.

\bibitem{tuamothesis}
Tuomo Lempiäinen.
\newblock {\em {Logic and Complexity in Distributed Computing}}.
\newblock PhD thesis, Aalto University, 2019.

\bibitem{libkin}
Leonid Libkin.
\newblock {\em Elements of Finite Model Theory}.
\newblock Texts in Theoretical Computer Science. An {EATCS} Series. Springer,
  2004.

\bibitem{reitericalp}
Fabian Reiter.
\newblock Asynchronous distributed automata: {A} characterization of the modal
  mu-fragment.
\newblock In Ioannis Chatzigiannakis, Piotr Indyk, Fabian Kuhn, and Anca
  Muscholl, editors, {\em 44th International Colloquium on Automata, Languages,
  and Programming, {ICALP} 2017}, volume~80 of {\em LIPIcs}, pages
  100:1--100:14. Schloss Dagstuhl - Leibniz-Zentrum f{\"{u}}r Informatik, 2017.

\bibitem{reiterthesis}
Fabian Reiter.
\newblock {\em Distributed Automata and Logic. (Automates Distribu{\'{e}}s et
  Logique)}.
\newblock PhD thesis, Sorbonne Paris Cit{\'{e}}, France, 2017.

\end{thebibliography}

\section{Appendix A}

\noindent
\textbf{Proof of Lemma \ref{lemma:cmsc_to_msc}:}

\begin{proof}
    Let $\Lambda$ be a program with the conditional iteration clauses 
    \[
    \begin{aligned}
        &X_1 \coloneq_{\varphi_{1,1}, \ldots, \varphi_{1,m_1}} \psi_{1,1};\ldots;\psi_{1,m_1}; \chi_1 \\
        &\vdots \\
        &X_k \coloneq_{\varphi_{k,1}, \ldots, \varphi_{k,m_k}} \psi_{k,1};\ldots;\psi_{k,m_k}; \chi_k, \\
    \end{aligned}
    \]
    where $\varphi_{i,j}$ is the
    is the $j$th condition of
    the $i$th rule and $\psi_{i,j}$ the corresponding consequence. All remaining iteration clauses are non-conditional. We define a strongly equivalent $\Pi$-program $\Lambda'$ of $\msc$ as follows. We keep all the terminal clauses and non-conditional iteration clauses as they are. Conditional iteration clauses will be modified. The appointed predicates will be chosen as in $\Lambda$.
    For every head predicate $X_i$ of a conditional clause, the new iteration clause will be a non-conditional rule that has the same effect as the conditional clause. 
    The non-conditional rule is defined recursively by specifying schemata $\theta_{(i,1)},\dots , \theta_{(i,m_i)}$ for each $X_i$.
    Intuitively, $\theta_{(i,1)}$ will cover
    the situation where 
    the truth of $X_i$ is determined by the backup or 
    the last condition-consequence pair. The schema $\theta_{(i,2)}$ will deal with the penultimate condition-consequence pair, $\theta_{(i,3)}$ with the condition-consequence pair just before the penultimate one, and so on. The final schema $\theta_{(i,m_i)}$ deals with the first condition.

To define $\theta_{(i,1)},\dots , \theta_{(i,m_i)}$,
we first let 
\[\theta_{(i,1)} \coloneqq (\varphi_{i, m_i} \land \psi_{i, m_i}) \lor (\neg \varphi_{i, m_i} \land \chi_i).\]
Supposing that we have
defined $\theta_{(i,j)}$ for $j < m_i$, we let 
\[
\theta_{(i,j+1)} \coloneqq (\varphi_{i, m_i-j} \land \psi_{i, m_{i}-j}) \lor (\neg \varphi_{i, m_i-j} \land \theta_{(i,j)}).
\]
The iteration clause for the head predicate $X_i$ is $X_i \coloneq \theta_{(i,m_i)}$.

We have now defined $\Lambda'$. It is straightforward to show that $\Lambda'$ is strongly equivalent to $\Lambda$ and also linear in the size of $\Lambda$ with the same maximum modal depth both in relation to terminal as well as iteration clauses.  
\end{proof}

\noindent 
\textbf{Proof of Lemma \ref{ridofsubindiceslemma}:} 
\begin{proof}
    To simulate a program of $\mpmsc$ correctly, we have to time some subprograms. Intuitively, we will define a head predicate for each subschema of the form $\Diamond_i \psi$ that appears in the program and replace these subschemata by corresponding head predicates in the original program. The rules for these new head predicates are timed such that the rule for $\Diamond_i \psi$ activates precisely when we are scanning the $i$th neighbour. Moreover, we manipulate all relevant rules to wait until our new rules have scanned through all possible identifiers and thus correctly checked for conditions with diamonds.

    Let $\Pi$ be a finite set of propositions and $\Lambda$ be a $\Pi$-program of $\mpmsc$, and recall that the modal depth is zero for terminal clauses and at most one for iteration clauses. Assume that the maximum subindex in a diamond that appears in $\Lambda$ is $I$. From $\Lambda$ we will construct a program $\Gamma$ of $\cmsc$ of size $\ordo(I+\abs{\Pi_1}+\abs{\Lambda})$. We begin by including a copy of $\Lambda$ in $\Gamma$. After this, we add---according to the recipe in sections \ref{clock} and \ref{Ordering neighbors by propositions}---a program for simulating multimodal diamonds with the head predicates $(M_{\abs{\Pi_1}}, \ldots, M_1)$ and $(N_I, \ldots, N_1)$ (note that we even use the same symbols as in the appendix above), including all the rules in the above appendix sections ``\emph{A formal program for the clock}'' and ``\emph{A formal program for simulating multimodal diamonds.}'' We also simultaneously replace every subschema $\Diamond_i \psi$ of $\mpmsc$ with the $\Diamond_i \psi$ macro, as described in section \ref{Ordering neighbors by propositions} (and more formally in the appendix above). From here onwards, all schemata with a $\Diamond_i$ will refer to these new macros. 

    We time $\Gamma$ as follows. We start by separating some subschemata of $\Gamma$ to new subprograms in the way described next. For every $i \in [I]$, let $S_i \subseteq \subs(\Gamma)$ be the set of schemata in $\subs(\Gamma)$ of the form $\Diamond_i \psi$. Note that the modal depth of the terminal clauses of $\Gamma$ is zero, so $\Diamond_i$ only appears in iteration clauses. We let $S_i \coloneqq \{\Diamond_i \psi_1, \ldots, \Diamond_i\psi_{k_i}\}$, where $k_i \coloneqq \abs{S_i}$. For every $i \in [I]$ and $j \in [k_i]$, we add a fresh head predicate $X_{(i,j)}$ to $\Gamma$ with the rules
    \[
    \begin{aligned}
        &X_{(i,j)}(0) \coloneq \bot &&X_{(i,j)} \coloneq_{\neg N_{i}} \Diamond_i \psi_j; X_{(i,j)},
    \end{aligned}
    \]
    where $N_i$ is described in section \ref{Ordering neighbors by propositions} (and more formally in the appendix above). The rule for $X_{(i,j)}$ ``tests'' the schema $\Diamond_i\psi_j$ in every round until and (importantly) including the round the neighbour $i$ has been scanned. At first we may get ``wrong'' truth values for $X_{(i,j)}$, but in the actual round when the neighbour $i$ is scanned, this becomes fixed.

    Next, we replace each $\Diamond_i \psi_j \in S_i$ in the program $\Gamma$ with the corresponding schema variable $X_{(i,j)}$. Now we have timed every $\Diamond_i$, but we still have to time other parts of the program $\Gamma$ by using $X_{\text{reset}}$ from the program for simulating multimodal diamonds as follows. Let $X \coloneq \psi$ be a non-conditional iteration clause of $\Gamma$ that was originally taken from $\Lambda$, i.e., the rule is not part of the clocks or the other auxiliary programs. We manipulate this clause by adding a new condition to it such that the iteration clause $X \coloneq \psi$ is transformed into $X \coloneq_{X_{\text{reset}}} \psi; X$. 
    Adding the condition $X_{\text{reset}}$ to this clause prevents possible misuse of the diamond $\Diamond_i$ macro, since the macro is only reliable at specific times. Now, let $Y \coloneq_{\varphi_1, \ldots, \varphi_k} \psi'_1; \ldots; \psi'_k; \chi$ be a \emph{conditional} iteration clause of $\Gamma$ that was originally taken from $\Lambda$. We similarly add a new condition to it such that
the rule is transformed into 
    $
    Y \coloneq_{\neg X_{\text{reset}}, \varphi_1, \ldots, \varphi_k} Y; \psi'_1; \ldots; \psi'_k; \chi.
    $
    Since the condition $X_{\text{reset}}$ is true just before the basic clock resets, we can be sure that all head predicates $N_i$ that would turn true have indeed already become true in the current or some earlier block. Note that the diamond macro for $\Diamond_i$ is not reliable in any block before the one where $N_i$ becomes true, and furthermore, even in the block where $N_i$ does become true, the first round is not reliable. 
Our $\cmsc$-program is now ready. Its size is clearly $\ordo(I+\abs{\Pi_1}+\abs{\Lambda})$. We explain briefly how the program $\Gamma$ simulates the head predicates in $\head(\Lambda)$. The terminal clauses for $\head(\Lambda)$ in $\Gamma$ have been kept the same as they were in $\Lambda$, so the beginning round for them works exactly the same as in $\Lambda$. Next, we describe how the rounds after the zeroeth round for $\Lambda$ are simulated by $\Gamma$, noting that simulating a single round, of course, lasts for several rounds. In a single cycle of simulation---simulating a round of $\Lambda$---the clock goes through all possible identifiers, and the cycle ends with $X_{\text{reset}}$ being true. Within the cycle, our new iteration clauses for the head predicates $X_{(i,j)}$ of $\Gamma$ are working with the clock and other auxiliary programs to simulate each $\Diamond_i$. When a $\Diamond_i$ is reliable, the new head predicates $X_{(i,j)}$ are active and store the truth values for $\Diamond_i \psi_j$. When the clock has scanned all possible identifiers, $X_{\text{reset}}$ becomes true, and then every original head predicate of $\head(\Lambda)$ in $\Gamma$ activates in such a way that it operates exactly as the corresponding rule in $\Lambda$. At the same time, the flags $N_i$ reset. A new cycle will then begin. It is relatively easy to more rigorously show that $\Gamma$ is (acceptance) equivalent to $\Lambda$.

    The ``time loss'' in every round can be seen from the definition of the clock for diamonds $\Diamond_i$. The delay depends on the length of the identifiers, since we always have to check all $\ID$s. It takes $\ordo(2^{\abs{\Pi_1}})$ rounds to simulate a round of $\Lambda$. 
\end{proof}

\noindent
\textbf{Proof of Lemma \ref{lem:term_zero}:}

\begin{proof}
    Let $\Lambda$ be a $\Pi$-program of $\msc$. If $\mdt(\Lambda) = 0$, we are done, so we may assume that $\mdt( \Lambda ) > 0$. Let $X_1, \ldots, X_q$ be the head predicates of $\Lambda$, and let $\varphi_1,\ldots, \varphi_q$ be the bodies of the terminal clauses and $\psi_1,\ldots, \psi_q$ the bodies of the iteration clauses. We define a program $\Lambda'$ of \emph{conditional} $\msc$ as follows.
    We add a new schema variable $I$ such that $I( 0 ) \coloneq \bot$, $I \coloneq \top$ to indicate whether we are computing an iteration step or the terminal step. We replace every head predicate $X_i$ with a head $X_{\psi_i}$ and set $X_{\psi_i}( 0 ) \coloneq \bot$ and $X_{\psi_i} \coloneq_I  \psi_i ; \varphi_i$. The attention (respectively, print) predicates are chosen in $\Lambda'$ such that $X_{\psi_i}$ is an attention (resp., print) predicate in $\Lambda'$ if and only if $X_i$ is an attention (resp., print) predicate in $\Lambda$. 

    It is straightforward to prove that 
    $
        ( M, w ) \models X^{n + 1}_{\psi_i} \iff ( M, w ) \models X^{n}_{i}  
    $ 
    holds for every pointed $\Pi$-model $( M, w ) $, $n \ge 0$ and $i = 1, \ldots, q$. 
    Also,  $\abs{\Lambda'} = \ordo( \abs{\Lambda} )$. The computation time of $\Lambda'$ is linear in the computation time of $\Lambda$.
%
%
%
\end{proof}

\noindent 
\textbf{Proof of Theorem \ref{thr:msc1}:}

\begin{proof}
    The transformation is based on labeling subschemata of the $\msc$-program with fresh schema variables. First, we briefly describe the proof idea informally.
    We define a clock that allows us to work with cycles of program evaluation steps, the length of a cycle being the maximum modal depth of the original program. The new program simulates the original one with evaluation formulas of increasing modal depth, step by step, until reaching the maximum modal depth. 
    
    Let $\Gamma$ be a $\Pi$-program of $\msc$. By Lemma \ref{lem:term_zero}, there exists an equivalent $\Pi$-program $\Lambda$ of $\msc$ with $\mdt(\Lambda) = 0$. If $\mdi( \Lambda ) \le 1$, we are done, so we assume that $\mdi( \Lambda ) > 1$. 
    Let $X_1, \ldots, X_q$ be the head predicates of $\Lambda$, and let $\varphi_1,\ldots, \varphi_q$ 
    (respectively, $\psi_1,\ldots, \psi_q$) be the bodies of the corresponding terminal (resp., iteration) clauses.    Let $S_{\Diamond}$ 
    be the set of subschemata of the schemata
    $\psi_1, \ldots, \psi_q$ of
    the form $\Diamond\psi$. 
%
%
%
    Let $S_i = \left\{ \, \Diamond \psi \in S_{\Diamond} \mid \md( \Diamond \psi ) = i \, \right\}$ for every $i = 1, \ldots,  \mdi( \Lambda )-1$, i.e., we split $S_{\Diamond}$ by modal depth into subsets. We let $k_i := |S_i|$ and let $\theta_{( i, 1 )} ,\ldots, \theta_{( i, k_i )}$ denote the schemata in $S_i$, i.e., 
    now $S_i = \{\theta_{( i, 1 )} ,\ldots, \theta_{( i, k_i )}\}$. We are now ready to translate $\Lambda$ to the desired program $\Lambda'$.

    We first define a clock that ticks for $\mdi(\Lambda)$ rounds and then repeats from the beginning. The program for the clock consists of the head predicates $T_1,\dots , T_{\mdi(\Lambda)}$ and the following rules: $T_1 ( 0 ) \coloneq \top$, $T_1 \coloneq T_{\mdi(\Lambda)}$ and for $i \in [\mdi(\Lambda) -1]$, the rules $T_{i+1} (0) \coloneq \bot$ and $T_{i+1} \coloneq T_i$. 

We will next build rules that evaluate the original program one modal depth at a time. 
\emph{After the construction, we will give a concrete example} which may help the reader already while reading the technical specification of the rules below.

Now, for each head predicate $X_i$ of
$\Lambda$ and the body $\varphi_i$ of
the corresponding terminal clause, 
we define a fresh head predicate $X_{X_i}$ and a corresponding terminal clause $X_{X_{i}} (0)\coloneq \varphi_i$. These are added to $\Lambda'$. The corresponding iteration clause will be defined later on. 
Before that, we define some auxiliary rules for subschemata $\theta_{(i,j)}$.

We add a fresh head predicate $X_{\theta( 1, j ) }$ to $\Lambda'$ with the rules 
    $
    X_{\theta( 1, j ) } ( 0 ) \coloneq \bot\text{ and }X_{\theta( 1, j ) } \coloneq_{T_1} \theta^*_{( 1, j )}; X_{\theta( 1, j ) }
    $
    for every $j \in \left[ k_0 \right] $, where $\theta^*_{( 1, j )}$ is the schema obtained from $\theta_{( 1, j )}$ by replacing every head predicate $X_i$ with $X_{X_i}$. 
For  $i \in \{2,\dots , \mdi( \Lambda ) - 1\}$, we add a fresh head predicate $X_{\theta( i, j ) }$ to $\Lambda'$ with the rules $X_{\theta( i, j ) } ( 0 ) \coloneq \bot$ and   $ 
        X_{\theta_{( i, j )} } \coloneq_{T_{i}} \theta^*_{( i, j )} ; X_{\theta_{( i, j )}}, 
    $ where the schema $\theta^*_{( i, j )}$ is defined from $\theta_{(i,j)}$ as follows. We start by replacing every $\theta_{(i-1,j')}$ in $\theta_{(i,j)}$ with $X_{\theta_{(i-1, j')}}$, meaning that we replace every modal depth $i-1$ subschema of $\theta_{(i,j)}$ with the corresponding head predicate. Next, in the obtained schema, we replace every subschema of modal depth $i-2$ with the corresponding head predicate. Then we continue by replacing modal depth $i - 3$ subschemata, and so on, until we have replaced every possible subschema with a corresponding head predicate. Last, we replace every head predicate $X_i$ with $X_{X_i}$. This ultimately gives the schema $\theta^*_{(i,j)}$ a modal depth of one.

    Next we define the iteration clauses for the head predicates $X_{X_i} $ in $\Lambda' $. For every $i \in \left[ k \right] $, we define the clause $X_{X_{i}} \coloneq_{T_{\mdi(\Lambda)}} \psi_{i}^*; X_{X_{i}}$, where the schema $\psi^*_{i}$ is defined from $\psi_i$ by using a similar strategy as for the schemata $\theta_{(i,j)}$, that is, we first replace every modal depth $\mdi(\Lambda)-1$ subschema $\theta \in S_{\mdi(\Lambda)-1}$ of $\psi_i$ with the corresponding head predicate $X_{\theta}$, then similarly for modal depth $\mdi(\Lambda)-2$, and so on. Last, we replace every head predicate $X_i$ with $X_{X_i}$.

    As a final step, the attention (respectively, print) predicates of $\Lambda'$ are chosen such that $X_{X_i}$ is an attention (resp, print) predicate in $\Lambda'$ if and only if $X_{i}$ is an attention (resp., print) predicate in $\Lambda$. 

    We give a simple example of the construction.  Consider the case where $\Lambda$ is the program
    \[
    \begin{aligned}
        &X(0) \coloneq p\qquad &&X \coloneq \Diamond\Diamond( \Diamond \Diamond X \land \Diamond X).
    \end{aligned}
    \]
    By the construction above, the program $\Lambda'$ contains---in addition to the clock rules, as described before---the terminal and iteration clauses
    \[
    \begin{aligned}
        &X_X(0) \coloneq p &&X_X \coloneq_{T_{4}} \Diamond X_{\Diamond(\Diamond \Diamond X \land \Diamond X)}; X_X \\
                &X_{\Diamond X}(0) \coloneq \bot
        &&X_{\Diamond X} \coloneq_{T_1} \Diamond X_{X}; X_{\Diamond X}.\\
        &X_{ \Diamond \Diamond X}(0) \coloneq \bot &&X_{\Diamond \Diamond X} \coloneq_{T_2} \Diamond X_{\Diamond X}; X_{\Diamond \Diamond X}\\
        &X_{\Diamond( \Diamond \Diamond X \land  \Diamond X)}(0) \coloneq \bot &&X_{\Diamond( \Diamond \Diamond X \land  \Diamond X)} \coloneq_{T_3} \Diamond (X_{\Diamond\Diamond X} \land X_{ \Diamond X}); X_{\Diamond(\Diamond \Diamond X \land \Diamond X)}
    \end{aligned}
    \]

    
    It is straightforward to show by induction that the following conditions hold.
    \begin{enumerate}
        \item For every pointed $\Pi$-model $(M,w)$, we have $( M, w ) \models \theta_{( i, j )}^{n} \Leftrightarrow ( M, w ) \models X^{\mdi(\Lambda)n + i}_{\theta_{( i, j )} }$ for every $n$.
        \item For any $n\in \mathbb{N}$, and for any computation round $m\in\mathbb{N}$ such that
        $\mdi( \Lambda) n + i < m < \mdi( \Lambda) (n+1) + i,$
        the variable $X_{\theta_{(i,j)}}$ gets the same interpretation in round $m$ as in the round $\mdi( \Lambda) n + i$.
        
\end{enumerate}
    
    Therefore, the following conditions hold.
\begin{enumerate}
\item
    For every pointed $\Pi$-model $( M, w ) $, we have
    $
        ( M, w ) \models X^n_{i} \Leftrightarrow ( M, w ) \models X^{\mdi( \Lambda) n}_{X_{i}}
    $
for every $n$. 
\item
For any $n\in \mathbb{N}$, and for any computation round $m\in\mathbb{N}$ such that
$\mdi( \Lambda) n < m < \mdi( \Lambda) (n+1),$
the variable $X_{X_i}$ gets the same interpretation in round $m$ as in round $\mdi( \Lambda) n$.
\end{enumerate}
Therefore, the translated program is equivalent to the original program $\Lambda$ and, therefore, equivalent to $\Gamma$. The size of $\Lambda'$ is clearly $\ordo(\abs{\Lambda}) = \ordo(\abs{\Gamma})$ since the number of subschemata of $\Lambda$ is linear in $|\Lambda|$. The computation time is $\ordo(max(1,\md(\Lambda)))$ times the computation time of $\Gamma$.
\end{proof}

\section{Appendix B}

\subsection{Forest decomposition}\label{Forest decomposition}

The program starts by simulating the separation of the graph into oriented forests. It is not sufficient for the nodes to know the order of their neighbors; they also need to know whether the neighbors have a lower or higher $\ID$ than the node itself. This is done by comparing the bits of the colors one by one, which takes $\ell$ iterations. We count these iterations with the variables $\mathrm{T}_1, \dots, \mathrm{T}_{\ell + 2}$ (the terminal clause for each variable is $\bot$ unless otherwise specified):
\[
\mathrm{T}_1 \coloneq \top,
\quad
\mathrm{T}_2 \coloneq \mathrm{T}_1,
\quad
\dots
\quad
\mathrm{T}_{\ell} \coloneq \mathrm{T}_{\ell - 1},
\quad
\mathrm{T}_{\ell + 1} \coloneq \mathrm{T}_{\ell},
\quad
\mathrm{T}_{\ell + 2} \coloneq \mathrm{T}_{\ell + 1}.
\]
Each variable starts off as untrue ($\mathrm{T}_i (0) \coloneq \bot$). The variable $\mathrm{T}_1$ becomes true after the first round of iteration, the variable $\mathrm{T}_2$ becomes true after the second round, and so forth.

For the purpose of comparing a node's $\mathrm{ID}$ with those of its neighbors, we start by defining variable symbols $\mathrm{I}_i$ that contain the bits of a node's $\ID$:
\[
\begin{aligned}
    &\mathrm{I}_1 (0) \coloneq p_1, & &\mathrm{I}_2 (0) \coloneq p_2, & &\dots, & &\mathrm{I}_{\ell} (0) \coloneq p_{\ell}. \\
    &\mathrm{I}_1 \coloneq_{\mathrm{T}_{\ell + 1}, \mathrm{T}_1} \mathrm{I}_1; \mathrm{I}_{\ell}; \mathrm{I}_1, & &\mathrm{I}_2 \coloneq_{\mathrm{T}_{\ell + 1}, \mathrm{T}_1} \mathrm{I}_2; \mathrm{I}_1; \mathrm{I}_2, & &\dots, & &\mathrm{I}_{\ell} \coloneq_{\mathrm{T}_{\ell + 1}, \mathrm{T}_1} \mathrm{I}_{\ell}; \mathrm{I}_{\ell - 1}; \mathrm{I}_{\ell}.
\end{aligned}
\]
In the terminal clauses, each variable $\mathrm{I}_i$ receives the truth value of the proposition $p_i$. In the iteration clause, the variables cycle through each bit of the node's $\ID$ after the first iteration round. Once all $\ell$ bits have visited the $\ell$th position, the flag $\mathrm{T}_{\ell + 1}$ halts the cycle, leaving each variable with its original value (same as the terminal clause).

We also define variables $\mathrm{I}^\delta_i$ ($1 \leq \delta \leq \Delta$) that hold the bits of a node's $\delta$th neighbor's $\ID$:
\[
    \mathrm{I}^{\delta}_1 \coloneq_{T_{\ell + 1}, T_1} \mathrm{I}^{\delta}_1; \mathrm{I}^{\delta}_{\ell}; \Diamond_{\delta} \mathrm{I}_1,
    \quad
    \mathrm{I}^{\delta}_2 \coloneq_{T_{\ell + 1}, T_1} \mathrm{I}^{\delta}_2; \mathrm{I}^{\delta}_1; \Diamond_{\delta} \mathrm{I}_2,
    \quad
    \dots,
    \quad
    \mathrm{I}^{\delta}_{\ell} \coloneq_{T_{\ell + 1}, T_1} \mathrm{I}^{\delta}_{\ell}; \mathrm{I}^{\delta}_{\ell - 1}; \Diamond_{\delta} \mathrm{I}_{\ell}.
\]
In the first iteration round, the variables $\mathrm{I}^{\delta}_i$ inherit the bits of the $\ID$ of the $\delta$th neighbor; this is the first global communication round. After this, the variables begin rotating just like the variables $\mathrm{I}_i$, and they stop rotating at the same time as well.

After one iteration round, we have the leftmost bit of a node's $\ID$ in variables $\mathrm{I}_{\ell}$ and the leftmost bits of its neighbor's $\ID$s in variables $\mathrm{I}^{\delta}_{\ell}$. After a second iteration round, these same variables will hold the second bit, and so forth. We determine which node has the higher $\ID$ by finding the first bit that separates the $\ID$s; the node whose bit is $1$ has the higher $\ID$.

Next, we define variables $\mathrm{DIF}^{\delta}$ ($1 \leq \delta \leq \Delta$) for comparing the bits of the $\ID$s:
\[
    \mathrm{DIF}^{\delta} \coloneq_{\mathrm{T}_1} \neg \left( \mathrm{I}_{\ell} \leftrightarrow \mathrm{I}^{\delta}_{\ell} \right); \bot.
\]
The variable $\mathrm{DIF}^{\delta}$ is true in node $v$ if and only if the bits of the colors of the node and its $\delta$th neighbor don't match in position $\ell$. No comparison occurs in the first iteration round.

We take advantage of the first round in which $\mathrm{DIF}^{\delta}$ becomes true to define variables $\mathrm{HIGH}^{\delta}$ and $\mathrm{LOW}^{\delta}$ ($1 \leq \delta \leq \Delta$) that tell the node which neighbors have a higher $\ID$ and which neighbors have a lower $\ID$, since this depends entirely on the first distinct bit we find:
\[
    \mathrm{HIGH}^{\delta} \coloneq_{\mathrm{HIGH}^{\delta}, \mathrm{LOW}^{\delta}, \mathrm{DIF}^{\delta}} \top; \bot; \mathrm{I}_1; \bot,
    \quad
    \mathrm{LOW}^{\delta} \coloneq_{\mathrm{LOW}^{\delta}, \mathrm{HIGH}^{\delta}, \mathrm{DIF}^{\delta}} \top; \bot; \neg \mathrm{I}_1; \bot.
\]
The variable $\mathrm{HIGH}^{\delta}$ becomes true in a node if it has a higher $\ID$ than its $\delta$th neighbor. Likewise, the variable $\mathrm{LOW}^{\delta}$ becomes true if it has a lower $\ID$. Once the distinct bits are found in the variables $\mathrm{I}_\ell$, we copy their values from the variables $\mathrm{I}_1$ in the next iteration. After $\ell + 2$ iteration rounds we reach the point where either $\mathrm{HIGH}^{\delta}$ or $\mathrm{LOW}^{\delta}$ is true for every $\delta$ in every node. For the sake of the next definition, we also define $\mathrm{HIGH}^{0} \coloneqq \top$.

It is now possible to define the following formula:
\[
    \langle \delta \rangle \varphi \coloneqq \bigvee\limits_{\delta \leq i \leq \Delta} \left( \mathrm{HIGH}^{i - \delta} \land \mathrm{LOW}^{i - \delta + 1} \land \Diamond_i \varphi \right).
\]
The formula $\langle \delta \rangle \varphi$ is true in node $v$ if and only if $\varphi$ is true in the parent of node $v$ in the oriented forest $F_{\delta}$. Finally, we define the formula $\mathrm{END1} \coloneqq \mathrm{T}_{\ell + 2}$ that ends this first phase.

\subsection{Cole-Vishkin}\label{Cole-Vishkin}

We extend the program thus far to simulate the CV algorithm in each oriented forest. We compare the colors of nodes bit by bit and place communication rounds between the color comparisons. To do this, we use the clock defined in section \ref{clock}, using $\mathrm{END1}$ to mark the moment the clock begins ticking. This means that the clock variables keep their terminal truth values until this phase begins. We use the same variable symbols as before: $\mathrm{S}_1, \dots, \mathrm{S}_{\log(\ell)}$, $\mathrm{M}_1, \dots, \mathrm{M}_{\log(\ell)}$ and $\mathrm{S}_{\text{changing}}$. For the sake of convenience, we define that $\mathrm{S}_{\text{changing}} (0) \coloneq \top$; this only postpones the clock by a single iteration. We also define the formulae $\mathrm{S}_\text{on} \coloneqq \mathrm{S}_{\text{changing}}$ and $\mathrm{S}_\text{off} \coloneqq \neg \mathrm{S}_\text{on}$ as helpful abbreviations. The strings of the clock have $\log(\ell) = \log\log(n)$ bits because the minute hand always refers to one of the $\log(n)$ bits of a color and codes its position into binary a 2nd time.

We require an additional flag to mark the moments when global communication rounds take place:
\[
    \mathrm{CR} \coloneq_{\mathrm{END1}} \bigwedge\limits_{i = 1}^{\log(\ell)} \neg \mathrm{M}_i \land \mathrm{S}_{\text{on}}; \bot.
\]
The variable $\mathrm{CR}$ uses a unique moment at the very beginning of the clock's cycle to mark the global communication rounds. The unique moment is the first of two iterations where the minute hand is all zeroes. In the second such iteration, the formula $\mathrm{S}_\text{on}$ is no longer true. The same unique moment could be defined with the help of the variable $X_{\text{reset}}$ defined in section \ref{Ordering neighbors by propositions}, but we have no other need for the necessary forward clock(s).

For comparing the colors between a node and its neighbors, we define variables $\mathrm{B}^{\delta}_i$ ($1 \leq \delta \leq \Delta$) for the bits of the color of a node in forest ${F_{\delta}}$, starting with its $\mathrm{ID}$. We also define variables $\mathrm{P}^{\delta}_i$ ($1 \leq \delta \leq \Delta$) for the bits of the color of the parent of a node in forest ${F_{\delta}}$:
\[
\begin{aligned}
    &\mathrm{B}^{\delta}_1 \coloneq_{\mathrm{CR}, \mathrm{S}_{\text{off}}} \mathrm{N}^{\delta}_1; \mathrm{B}^{\delta}_2; \mathrm{B}^{\delta}_1,
    & &\mathrm{B}^{\delta}_2 \coloneq_{\mathrm{CR}, \mathrm{S}_{\text{off}}} \mathrm{N}^{\delta}_2; \mathrm{B}^{\delta}_3; \mathrm{B}^{\delta}_2,
    & &\dots,
    & &\mathrm{B}^{\delta}_{\ell} \coloneq_{\mathrm{CR}, \mathrm{S}_{\text{off}}} \mathrm{N}^{\delta}_{\ell}; \mathrm{B}^{\delta}_1; \mathrm{B}^{\delta}_{\ell}. \\
    &\mathrm{P}^{\delta}_1 \coloneq_{\mathrm{CR}, \mathrm{S}_{\text{off}}} \langle \delta \rangle \mathrm{N}^{\delta}_1; \mathrm{P}^{\delta}_2; \mathrm{P}^{\delta}_1,
    & &\mathrm{P}^{\delta}_2 \coloneq_{\mathrm{CR}, \mathrm{S}_{\text{off}}} \langle \delta \rangle \mathrm{N}^{\delta}_2; \mathrm{P}^{\delta}_3; \mathrm{P}^{\delta}_2,
    & &\dots,
    & &\mathrm{P}^{\delta}_{\ell} \coloneq_{\mathrm{CR}, \mathrm{S}_{\text{off}}} \langle \delta \rangle \mathrm{N}^{\delta}_{\ell}; \mathrm{P}^{\delta}_1; \mathrm{P}^{\delta}_{\ell}.
\end{aligned}
\]
Each variable $\mathrm{B}^{\delta}_i$ cycles through the bits of the color of the node in forest $F_{\delta}$. Each variable $\mathrm{P}^{\delta}_i$ does the same for the bits of the color of its parent. The variables rotate once every ``minute'', marked by the condition $\mathrm{S}_{\text{off}}$. The minute hand always indicates the original position of the bits currently in $\mathrm{B}^{\delta}_1$ and $\mathrm{P}^{\delta}_1$ in binary. During the global communication round, the variables $\mathrm{B}^{\delta}_i$ update to the next color $\mathrm{N}^{\delta}_i$ (see below) and the variables $\mathrm{P}^{\delta}_i$ receive the new color of the parent of the node in forest $F_\delta$. The formulae $\langle \delta \rangle \mathrm{N}^{\delta}_i$ are always untrue for root nodes, which means that the root compares its color to the color $0$, just as described in section \ref{Summary of the Cole-Vishkin algorithm}.

We compare the bits of colors between child and parent nodes using the same variables $\mathrm{DIF}^{\delta}$ as before, using the flag $\mathrm{CR}$ to prevent false positives. We compare bits in position $1$ instead of $\ell$ like before, because the CV algorithm seeks for the rightmost distinct bit:
\[
    \mathrm{DIF}^{\delta} \coloneq_{\mathrm{CR}, \mathrm{END1}, \mathrm{T}_1} \bot; \neg \left( \mathrm{B}^{\delta}_1 \leftrightarrow \mathrm{P}^{\delta}_1 \right); \neg \left( \mathrm{I}_{\ell} \leftrightarrow \mathrm{I}^{\delta}_{\ell} \right); \bot.
\]
We also define variables $\mathrm{GET}^{\delta}$ ($1 \leq \delta \leq \Delta$) that become true after the first distinct bit is found, until they are reset during the global communication rounds:
\[
    \mathrm{GET}^{\delta} \coloneq_{\mathrm{CR}, \mathrm{GET}^{\delta}} \bot; \top; \mathrm{DIF}^{\delta}.
\]

Now it is possible to define the variables $\mathrm{N}^{\delta}_i$ ($1 \leq \delta \leq \Delta$) that compute the next color of a node during each round of the clock:
\[
\begin{aligned}
    &\mathrm{N}^{\delta}_1 \coloneq_{\mathrm{GET}^{\delta}, \mathrm{DIF}^{\delta}} \mathrm{N}^{\delta}_1; \mathrm{B}^\delta_1; p_1, \\
    &\mathrm{N}^{\delta}_2 \coloneq_{\mathrm{GET}^{\delta}, \mathrm{DIF}^{\delta}} \mathrm{N}^{\delta}_2; \mathrm{M}_1; p_2, \\
    &\vdots \\
    &\mathrm{N}^{\delta}_{\ell} \coloneq_{\mathrm{GET}^{\delta}, \mathrm{DIF}^{\delta}} \mathrm{N}^{\delta}_{\ell}; \mathrm{M}_{\ell - 1}; p_{\ell}.
\end{aligned}
\]
The variables hold the $\ID$ of a node by default, which is used by the variables $\mathrm{B}^\delta_i$ and $\mathrm{P}^\delta_i$ during the first global communication round. The condition $\mathrm{DIF}^{\delta}$ copies the truth value of the first distinct bit to the variable $\mathrm{N}^{\delta}_1$ and the bits of its position from the minute hand to the rest of the variables $\mathrm{N}^{\delta}_i$. The condition $\mathrm{GET}^{\delta}$ preserves the calculated color until the next round of the clock, ensuring that only the first distinct bit is considered.

The CV algorithm takes $L = \log^*(n) + 3$ communication rounds. We count them with an hour hand consisting of variables $\mathrm{H}_1, \dots, \mathrm{H}_L$; after $\log^*(n) + 2$ ``hours'', $\mathrm{H}_L$ becomes true:
\[
    \mathrm{H}_1 \coloneq_{\mathrm{H}_1 \lor \mathrm{CR}} \top; \bot,
    \quad
    \mathrm{H}_2 \coloneq_{\mathrm{H}_2 \lor (\mathrm{H}_1 \land \mathrm{CR})} \top; \bot,
    \quad
    \dots,
    \quad
    \mathrm{H}_L \coloneq_{\mathrm{H}_L \lor (\mathrm{H}_{L - 1} \land \mathrm{CR})} \top; \bot.
\]
After one more hour, the clock is stopped using the following variable:
\[
\mathrm{STOP} \coloneq \mathrm{H}_L \land \left( \bigwedge\limits_{i = 1}^{\log(\ell)} \neg \mathrm{M}_i \land \mathrm{S}_{\text{on}} \right)
\]
The formula contained in the parentheses is the same one used for $\mathrm{CR}$. Thus, the variable turns on at the same time as $\mathrm{CR}$ in the last ``hour'' and can be used to stop the iteration of the variables $\mathrm{P}^{\delta}_{i}$ to prevent an unnecessary global communication round. One iteration later, the flag $\mathrm{END2} \coloneq \mathrm{STOP}$ is used to halt all other previously mentioned variables.

For the sake of Lemma \ref{Cole-Vishkin lemma}, we define $8^\Delta$ ``appointed'' predicates. These predicates correspond to the color of each node when its $\Delta$ colors from the forests are concatenated:
\[
\begin{aligned}
    &\mathrm{CLR}_0 \coloneq_{\mathrm{END2}} \bigwedge\limits_{\delta = 1}^{\Delta} \left( \neg \mathrm{B}^\delta_3 \land \neg \mathrm{B}^\delta_2 \land \neg \mathrm{B}^\delta_1 \right); \bot, \\
    &\vdots \\
    &\mathrm{CLR}_{8^{\Delta} - 1} \coloneq_{\mathrm{END2}} \bigwedge\limits_{\delta = 1}^{\Delta} \left( \mathrm{B}^\delta_3 \land \mathrm{B}^\delta_2 \land \mathrm{B}^\delta_1 \right); \bot.
\end{aligned}
\]
The colors are simply assembled from the bits. We can remove the variables $\mathrm{CLR}_0$, $\mathrm{CLR}_8$, $\mathrm{CLR}_{16}$, $\dots$ from the program, because they correspond to concatenations where the color of a node is $0$ in some oriented forest, which is impossible. If we consider the remaining variables $\mathrm{CLR}_{i}$ to be appointed predicates, Lemma \ref{Cole-Vishkin lemma} is now true:

\bigskip 

\noindent
$\blacktriangleright$ \textbf{Lemma\ \ref{Cole-Vishkin lemma}.}\ \emph{There exists a formula of $\mpmsc$ with the following properties.}
\begin{enumerate}
    \item \emph{It defines a proper $7^\Delta$-coloring.}
    \item \emph{Ignoring appointed predicates, it has $\mathcal{O}(\Delta \log(n))$ heads and a size of $\mathcal{O}(\Delta^2 \log(n))$.}
    \item \emph{The number of iterations needed is $\mathcal{O}(\log(n) \log\log(n) \log^*(n))$.}
    \item \emph{There are exactly $\log^*(n) + 4$ global communication rounds.}
\end{enumerate}

\begin{proof}
1) We have $7^\Delta$ appointed (attention and print) predicates: $\mathrm{CLR}_0, \dots, \mathrm{CLR}_{8^\Delta - 1}$, minus the variables \[\mathrm{CLR}_0, \mathrm{CLR}_8, \mathrm{CLR}_{16}, \dots.\] Because each possible combination of the truth values of $\mathrm{B}^{\delta}_{1}$, $\mathrm{B}^{\delta}_{2}$ and $\mathrm{B}^{\delta}_{3}$ ($1 \leq \delta \leq \Delta$) corresponds to a single variable $\mathrm{CLR}_{i}$, the print predicates are mutually exclusive, meaning that only one of them can be true in a given node. Each of the omitted predicates corresponds to an impossible combination of truth values, which means that one of the print predicates must become true in each node. This means that each node outputs a $7^\Delta$-bit string, where exactly one bit is a $1$. Thus, the program defines a $7^\Delta$-coloring.

If two neighbors shared the same output, they would have the same combination of truth values of variables $\mathrm{B}^{\delta}_{1}$, $\mathrm{B}^{\delta}_{2}$ and $\mathrm{B}^{\delta}_{3}$ ($1 \leq \delta \leq \Delta$). Each variable $\mathrm{B}^{\delta}_{i}$ ($i > 3$, $1 \leq \delta \leq \Delta$) has become untrue in each node by this point, which means that the neighbors would share the same color in every forest. This is impossible, because there is some forest $F_{\delta}$ where one of the neighbors is the parent of the other, and the program ensures that a parent and child never share the same color. Thus, no two neighbors share the same output, and the coloring is proper.

2) Let us count the number of heads and their lengths:
\begin{itemize}
    \item We have $\mathcal{O}(1)$ variables $\mathrm{S}_{\text{changing}}$, $\mathrm{CR}$ and $\mathrm{STOP}$ of size $\mathcal{O}(\log\log(n))$.
    \item We have $\mathcal{O}(\Delta) \mathcal{O}(\log(n))$ variables $\mathrm{T}_i$, $\mathrm{I}_i$, $\mathrm{I}^\delta_i$, $\mathrm{DIF}^{\delta}$, $\mathrm{HIGH}^{\delta}$, $\mathrm{LOW}^{\delta}$, $\mathrm{S}_i$, $\mathrm{M}_i$, $\mathrm{B}^{\delta}_{i}$, $\mathrm{P}^{\delta}_{i}$, $\mathrm{GET}^{\delta}$, $\mathrm{N}^{\delta}_{i}$, $\mathrm{H}_i$ and $\mathrm{END2}$ of size $\mathcal{O}(\Delta)$.
\end{itemize}
Adding the heads together we get
\[
    \mathcal{O}(1) + \mathcal{O}(\Delta) \mathcal{O}(\log(n)) = \mathcal{O}(\Delta \log(n)).
\]
Multiplying the heads by the sizes we get
\[
\begin{aligned}
    &\mathcal{O}(1) \mathcal{O}(\log\log(n)) + \mathcal{O}(\Delta) \mathcal{O}(\log(n)) \mathcal{O}(\Delta) & &= & &\mathcal{O}(\log\log(n)) + \mathcal{O}(\Delta^2 \log(n)) \\
    & & &= & &\mathcal{O}(\Delta^2 \log(n)).
\end{aligned}
\]

3) The first phase of the algorithm until the variable $\mathrm{END1}$ becomes true lasts for $\mathcal{O}(\log(n))$ iterations. After this, the clock begins ticking and lasts for $\mathcal{O}(\log^*(n))$ rounds.

Each round of the clock has the same number of iterations. For each position of the minute hand, we have $\mathcal{O}(\log\log(n))$ iterations: the first iteration where the minute hand updates to a new position, at most $\log\log(n)$ iterations where the second hand counts the rightmost $1$s of the minute hand and one more iteration where the second hand stops moving. The minute hand has a total of $\log(n)$ positions and the clock runs for $\mathcal{O}(\log^*(n))$ ``hours''. This dwarfs the iterations of the first phase and gives us a time complexity of $\mathcal{O}(\log(n) \log\log(n) \log^*(n))$.

4) The variables $\mathrm{I}^\delta_i$ and $\mathrm{P}^{\delta}_{i}$ are the only ones to contain a diamond. These diamonds are activated in the first iteration round and in every round where the variable $\mathrm{CR}$ is true. The latter happens once every ``hour'' and the hour hand limits the number of ``hours'' to $\log^*(n) + 3$. Hence, there are exactly $\log^*(n) + 4$ global communication rounds.
\end{proof}

The first global communication round is used to receive the $\ID$s of neighbors for comparison; after this, the remaining $\log^*(n) + 3$ global communication rounds are used for the computation of CV, which matches the number of communication rounds of the CV algorithm.

\subsection{Shift-down}\label{Shift-down}

The shift-down phase of the algorithm picks up from the $7$-colorings where the previous phase left off. It works by repeating the following three steps in each oriented forest:
\begin{enumerate}
    \item Each node inherits the color of its parent (a root changes color to $1$ or $2$, distinct from its previous color). This makes it so that all the children of a node share the same color. Each node stores its previous color for comparison. 
    \item Each node receives as a message the color of its parent. This is also stored for comparison.
    \item Each node compares its own color (inherited in step 1), the colors of its children (stored in step 1) and the color of its parent (stored in step 2). If it has the greatest color, it changes into the smallest color $1$, $2$ or $3$ that is not shared by its children or parent.
\end{enumerate}
After each repeat of these steps, the highest color in each oriented forest is reduced by at least one. After $4$ repeats, the $7$-coloring in each oriented forest is reduced to a $3$-coloring. An optimal $2$-coloring can't be achieved with this algorithm.

We extend the program from before to include the shift-down technique, omitting the previous appointed predicates. We mark the current step of shift-down using three variables:
\[
    \mathrm{Z}_1 \coloneq_{\mathrm{Z}_1 \lor \mathrm{Z}_2, \mathrm{END2}} \bot; \top; \bot,
    \quad
    \mathrm{Z}_2 \coloneq_{\mathrm{Z}_1} \top; \bot,
    \quad
    \mathrm{Z}_3 \coloneq_{\mathrm{Z}_2} \top; \bot.
\]
All three variables are untrue until the end of the previous phase activates $\mathrm{Z}_1$. The variables take turns turning on and off; just one of them is true in any given iteration round. The lower index of the true variable indicates the step the program performs in the next iteration.

The variables only have to cycle through $4$ loops. To terminate this phase, we define another set of timer variables like we did in the first phase:
\[
    \mathrm{T}'_1 \coloneq_{\mathrm{T}'_1\lor \mathrm{Z}_3} \top; \bot,
    \quad
    \mathrm{T}'_2 \coloneq_{\mathrm{T}'_2 \lor (\mathrm{T}'_1 \land \mathrm{Z}_3)} \top; \bot,
    \quad
    \mathrm{T}'_3 \coloneq_{\mathrm{T}'_3 \lor (\mathrm{T}'_2 \land \mathrm{Z}_3)} \top; \bot,
    \quad
    \mathrm{T}'_4 \coloneq_{\mathrm{T}'_4 \lor (\mathrm{T}'_3 \land \mathrm{Z}_3)} \top; \bot.
\]
The variable $\mathrm{T}'_1$ becomes true after the first loop, $\mathrm{T}'_2$ becomes true after the second loop and so forth. Each variable remains true once it activates. We also define the formula $\mathrm{END3} \coloneqq \mathrm{T}'_4$, which marks the end of the shift-down phase.

Given that we have reduced the number of colors in each forest to just $7$, it is convenient to use variables $\mathrm{C}^\delta_i$ that refer directly to a node's color in each forest; we define them later. We also use variables $\mathrm{Cc}^{\delta}_{i}$ ($1 \leq \delta \leq \Delta$) that store a node's earlier color that is passed down to its children in step $1$ and variables $\mathrm{Cp}^{\delta}_{i}$ ($1 \leq \delta \leq \Delta$) that store the color of its parent in step $2$. The upper indices refer to the labels of the oriented forests and the lower indices refer to colors. The latter two sets of variables are easy to define:
\[
\begin{aligned}
    &\mathrm{Cc}^{\delta}_{1} \coloneq_{\mathrm{Z}_1} \mathrm{C}^{\delta}_{1}; \mathrm{Cc}^{\delta}_{1},
    & &\dots,
    & &\mathrm{Cc}^{\delta}_{7} \coloneq_{\mathrm{Z}_1} \mathrm{C}^{\delta}_{7}; \mathrm{Cc}^{\delta}_{7}. \\
    &\mathrm{Cp}^{\delta}_{1} \coloneq_{\mathrm{Z}_2} \langle \delta \rangle \mathrm{C}^{\delta}_{1}; \mathrm{Cp}^{\delta}_{1},
    & &\dots,
    & &\mathrm{Cp}^{\delta}_{7} \coloneq_{\mathrm{Z}_2} \langle \delta \rangle \mathrm{C}^{\delta}_{7}; \mathrm{Cp}^{\delta}_{7}.
\end{aligned}
\]

For the variables $\mathrm{C}^{\delta}_{i}$, we need to consider that the roots need to change their color differently in step 1. For this purpose, we define additional variables $\mathrm{R}^{\delta}$ ($1 \leq \delta \leq \Delta$) to flag the roots of each oriented forest:
\[
\mathrm{R}^{\delta} \coloneq_{\mathrm{R}^{\delta}, \mathrm{CR}} \top; \neg \langle \delta \rangle \top; \bot.
\]
A node is flagged as a root in forest $F_{\delta}$ only if it has no parent in that forest, i.e. if the diamond $\langle \delta \rangle$ doesn't lead anywhere. Since other variables already use diamonds when $\mathrm{CR}$ is true in the previous phase, this does not add any additional global communication rounds, nor do we need any additional flags.

For step 3, we use helpful formulae to condense things. We define formulae $\mathrm{G}^{\delta}$ ($1 \leq \delta \leq \Delta$), that are true when a node has a color greater than its children and parent in forest $F_{\delta}$:
\[
\mathrm{G}^{\delta} \coloneqq \mathrm{Z}_3 \land \bigvee_{i = 3}^{7} \left( \mathrm{C}^{\delta}_{i} \land \neg \bigvee_{j = i + 1}^{7} \left( \mathrm{Cc}^{\delta}_{j} \lor \mathrm{Cp}^{\delta}_{j} \right) \right).
\]
The formulae state that the program is in step 3, the node has a color between $3$ and $7$ in forest $F_\delta$, and its parent and children do not posses a greater color.

We also define formulae $\mathrm{L}^{\delta}_{i}$ ($1 \leq \delta \leq \Delta$). The lower index of the true formula tells us which of the colors $1$, $2$ and $3$ is the lowest available color a node can switch to in forest $F_{\delta}$:
\[
    \mathrm{L}^{\delta}_1 \coloneqq \neg \left( \mathrm{Cc}^{\delta}_{1} \lor \mathrm{Cp}^{\delta}_{1} \right),
    \quad
    \mathrm{L}^{\delta}_2 \coloneqq \neg \mathrm{L}^{\delta}_1 \land \neg \left( \mathrm{Cc}^{\delta}_{2} \lor \mathrm{Cp}^{\delta}_{2} \right),
    \quad
    \mathrm{L}^{\delta}_3 \coloneqq \neg \mathrm{L}^{\delta}_1 \land \neg \mathrm{L}^{\delta}_2.
\]

Now we are ready to define the variables $\mathrm{C}^{\delta}_{i}$ ($1 \leq \delta \leq \Delta$):
\[
\begin{aligned}
    &\mathrm{C}^{\delta}_{1}\coloneq_{\mathrm{G}^{\delta}, \mathrm{Z}_3 \lor \mathrm{Z}_2, \mathrm{R}^{\delta} \land \mathrm{Z}_1, \mathrm{Z}_1, \mathrm{END2}} & &\mathrm{L}^{\delta}_1; & &\mathrm{C}^{\delta}_{1}; & &\neg \mathrm{C}^{\delta}_{1}; & &\langle \delta \rangle \mathrm{C}^{\delta}_{1}; & &\neg \mathrm{B}^{\delta}_{3} \land \neg \mathrm{B}^{\delta}_{2} \land \mathrm{B}^{\delta}_{1}; & &\bot, \\
    &\mathrm{C}^{\delta}_{2} \coloneq_{\mathrm{G}^{\delta}, \mathrm{Z}_3 \lor \mathrm{Z}_2, \mathrm{R}^{\delta} \land \mathrm{Z}_1, \mathrm{Z}_1, \mathrm{END2}} & &\mathrm{L}^{\delta}_2; & &\mathrm{C}^{\delta}_{2}; & &\mathrm{C}^{\delta}_{1}; & &\langle \delta \rangle \mathrm{C}^{\delta}_{2}; & &\neg \mathrm{B}^{\delta}_{3} \land \mathrm{B}^{\delta}_{2} \land \neg \mathrm{B}^{\delta}_{1}; & &\bot, \\
    &\mathrm{C}^{\delta}_{3}\coloneq_{\mathrm{G}^{\delta}, \mathrm{Z}_3 \lor \mathrm{Z}_2, \mathrm{Z}_1, \mathrm{END2}} & &\mathrm{L}^{\delta}_3;  & &\mathrm{C}^{\delta}_{3}; & & & &\langle \delta \rangle \mathrm{C}^{\delta}_{3}; & &\neg \mathrm{B}^{\delta}_{3} \land \mathrm{B}^{\delta}_{2} \land \mathrm{B}^{\delta}_{1}; & &\bot, \\
    &\mathrm{C}^{\delta}_{4} \coloneq_{\mathrm{G}^{\delta}, \mathrm{Z}_3 \lor \mathrm{Z}_2, \mathrm{Z}_1, \mathrm{END2}} & &\bot; & &\mathrm{C}^{\delta}_{4}; & & & &\langle \delta \rangle \mathrm{C}^{\delta}_{4}; & &\mathrm{B}^{\delta}_{3} \land \neg \mathrm{B}^{\delta}_{2} \land \neg \mathrm{B}^{\delta}_{1}; & &\bot, \\
    &\vdots \\
    &\mathrm{C}^{\delta}_{7} \coloneq_{\mathrm{G}^{\delta}, \mathrm{Z}_3 \lor \mathrm{Z}_2, \mathrm{Z}_1, \mathrm{END2}} & &\bot; & &\mathrm{C}^{\delta}_{7}; & & & &\langle \delta \rangle \mathrm{C}^{\delta}_{7}; & &\mathrm{B}^{\delta}_{3} \land \mathrm{B}^{\delta}_{2} \land \mathrm{B}^{\delta}_{1}; & &\bot.
\end{aligned}
\]
The upper index refers to the label of the oriented forest and the lower index refers to the node's color in said forest. From left to right, we first consider step 3: if the node has a color greater than its children and parent, then the lowest available variable $\mathrm{C}^{\delta}_{1}$, $\mathrm{C}^{\delta}_{2}$ or $\mathrm{C}^{\delta}_{3}$ becomes true. If the node has a lower color or if we are in step 2, the variables don't update. If we are in step 1 and the node is a root, then only the variables $\mathrm{C}^{\delta}_{1}$ and $\mathrm{C}^{\delta}_{2}$ alternate ($\langle \delta \rangle$-formulae are always untrue for root nodes). If the node is not a root in step 1, then it inherits the color of its parent with the diamond $\langle \delta \rangle$. When the variables are activated, they are assembled from the bit configuration of the final color in the previous phase.

Finally, we can use the flag $\mathrm{END3}$ as a condition to stop all the variables in this section from updating after $4$ loops of steps 1-3. One iteration before the formula $\mathrm{END3}$ becomes true, the variable $\mathrm{Z}_3$ is true, which means that the color variables $\mathrm{C}^{\delta}_{i}$ finish their current loop before $\mathrm{END3}$ stops them on the next iteration.

Once again, we define $4^\Delta$ ``appointed'' predicates for the sake of a helpful lemma. They signify the final color of a node in oriented forest $F_{\delta}$:
\[
\begin{aligned}
    &\mathrm{CLR}_{0} \coloneq_{\mathrm{END3}} \bigwedge\limits_{\delta = 1}^{\Delta} \left( \neg \mathrm{C}^{\delta}_{3} \land \neg \mathrm{C}^{\delta}_{2} \land \neg \mathrm{C}^{\delta}_{1} \right); \bot, \\
    &\vdots \\
    &\mathrm{CLR}_{4^\Delta - 1} \coloneq_{\mathrm{END3}} \bigwedge\limits_{\delta = 1}^{\Delta} \left( \mathrm{C}^{\delta}_{3} \land \neg \mathrm{C}^{\delta}_{2} \land \neg \mathrm{C}^{\delta}_{1} \right); \bot.
\end{aligned}
\]
The colors are assembled from the forest-specific colors as if they were pairs of bits. As such, we can once again remove the variables $\mathrm{CLR}_{0}, \mathrm{CLR}_{4}, \mathrm{CLR}_{8}, \dots$ from the program, because the color $0$ is not used in any forest. If we consider the remaining variables $\mathrm{C}^{\delta}_{i}$ to be appointed predicates, the following lemma is true:

\begin{lemma}\label{Shift-down lemma}
There exists a formula of $\mpmsc$ with the following properties.
\begin{enumerate}
    \item It defines a proper $3^\Delta$-coloring.
    \item Ignoring appointed predicates, it has $\mathcal{O}(\Delta \log(n))$ heads and a size of $\mathcal{O}(\Delta^2 \log(n))$.
    \item The number of iterations needed is $\mathcal{O}(\log(n) \log\log(n) \log^*(n))$.
    \item There are exactly $\log^*(n) + 12$ global communication rounds.
\end{enumerate}
\end{lemma}

\begin{proof}
1) We have $3^\Delta$ appointed (attention and print) predicates: $\mathrm{CLR}_{0}, \dots, \mathrm{CLR}_{4^\Delta - 1}$, minus the variables \[\mathrm{CLR}_{0}, \mathrm{CLR}_{4}, \mathrm{CLR}_{8}, \dots.\] Each appointed predicate corresponds with a single true variable $\mathrm{C}^{\delta}_{1}$, $\mathrm{C}^{\delta}_{2}$ or $\mathrm{C}^{\delta}_{3}$ for each $1 \leq \delta \leq \Delta$, which means that they are once again mutually exclusive and only one of them can be true in a given node. For each $1 \leq \delta \leq \Delta$, the variables $\mathrm{C}^{\delta}_{i}$ ($1 \leq i \leq 7$) are mutually exclusive and one of them must also be true. Additionally, the variables $\mathrm{C}^{\delta}_{i}$ ($i > 3$) are all untrue at this point. This means that at least one appointed predicate must become true in each node. Thus, each node outputs a $3^\Delta$-bit string where exactly one bit is a $1$, and the program defines a $3^\Delta$-coloring.

If two neighbors shared the same output string, they would share the same truth values for all variables $\mathrm{C}^{\delta}_{i}$ ($1 \leq i \leq 7$, $1 \leq \delta \leq \Delta$). This means that the neighbors would share the same color in each oriented forest. This is once again impossible, because there is some oriented forest $F_\delta$ where one of the neighbors is the other's parent, and the program ensures that a parent and child never share the same color. Thus, no two neighbors share the same output, and the coloring is proper.

2) Let us count the number of heads and their lengths. According to Lemma \ref{Cole-Vishkin lemma}, we have $\mathcal{O}(\Delta \log(n))$ heads from before. Their size was $\mathcal{O}(\Delta^2 \log(n))$. Let us count the new heads and their sizes:
\begin{itemize}
    \item We have $\mathcal{O}(\Delta)$ variables $\mathrm{Z}_i$, $\mathrm{T}'_i$, $\mathrm{Cc}^{\delta}_{i}$, $\mathrm{Cp}^{\delta}_{i}$, $\mathrm{R}^{\delta}$ and $\mathrm{C}^{\delta}_{i}$ of size $\mathcal{O}(\Delta)$.
\end{itemize}
Adding the heads together with the previous part of the program, we get
\[
    \mathcal{O}(\Delta \log(n)) + \mathcal{O}(\Delta) = \mathcal{O}(\Delta \log(n)).
\]
Adding the sizes together, we get
\[
    \mathcal{O}(\Delta^2 \log(n)) + \mathcal{O}(\Delta) \mathcal{O}(\Delta) = \mathcal{O}(\Delta^2 \log(n)) + \mathcal{O}(\Delta^2) = \mathcal{O}(\Delta^2 \log(n)).
\]

3) The number of previous iterations was $\mathcal{O}(\log(n) \log\log(n) \log^*(n))$ by Lemma \ref{Cole-Vishkin lemma}. The number of iterations for the shift-down phase is the number of steps multiplied by the number of loops. Since there are $\mathcal{O}(1)$ steps and $\mathcal{O}(1)$ loops, this means an additional $\mathcal{O}(1)$ iterations. The total number of iterations is still
\[
\mathcal{O}(\log(n) \log\log(n) \log^*(n)) + \mathcal{O}(1) = \mathcal{O}(\log(n) \log\log(n) \log^*(n)).
\]

4) The previous phase had exactly $\log^*(n) + 4$ global communication rounds. In this phase, diamonds are used exactly during steps 1 and 2. Given that we go through $4$ loops of these steps, this gives us an additional $4 \cdot 2 = 8$ global communication rounds, which matches the algorithm. The total number of global communication rounds is then
\[
\log^*(n) + 4 + 8 = \log^*(n) + 12.
\]
\end{proof}

\subsection{Basic color reduction}\label{basic color reduction}

In the next phase, the $3$-colorings from the previous phase are combined together into a coloring for the whole graph by placing them together: the color of a node in forest $F_1$ becomes the two rightmost bits of its color, its color in forest $F_2$ becomes the next two bits, and so forth, until its color in forest $F_{\Delta}$ becomes the two leftmost bits. This gives us a proper $(3^{\Delta})$-coloring of the whole graph.

Finally, this is turned into a $(\Delta + 1)$-coloring with a technique called basic color reduction. In each round, the nodes send their color to all of their neighbors. After this, the nodes with a higher color than their neighbors change color to the lowest available color from the set $[\Delta + 1]$; at least one of these colors is available, because each node has at most $\Delta$ neighbors. We repeat this algorithm until the highest color of a node is $\Delta + 1$, which will take at most $7^{\Delta} - (\Delta + 1)$ communication rounds.

We need $2 \Delta$ variables that will contain the bits of the new color. For each $\delta$ we define two variables that contain the bits of a node's color in forest $F_\delta$. Together, these bits will define the whole color of the node in the graph:
\[
\begin{aligned}
    &\mathrm{B}_1 \coloneq_\mathrm{END3} \mathrm{C}^1_1 \lor \mathrm{C}^1_3; \bot,
    & &\mathrm{B}_3 \coloneq_\mathrm{END3} \mathrm{C}^2_1 \lor \mathrm{C}^2_3; \bot,
    & &\dots,
    & &\mathrm{B}_{2 \Delta - 1} \coloneq_\mathrm{END3} \mathrm{C}^{\Delta}_1 \lor \mathrm{C}^{\Delta}_3; \bot, \\
    &\mathrm{B}_2 \coloneq_\mathrm{END3} \mathrm{C}^1_2 \lor \mathrm{C}^1_3; \bot,
    & &\mathrm{B}_4 \coloneq_\mathrm{END3} \mathrm{C}^2_2 \lor \mathrm{C}^2_3; \bot,
    & &\dots,
    & &\mathrm{B}_{2 \Delta} \coloneq_\mathrm{END3} \mathrm{C}^{\Delta}_2 \lor \mathrm{C}^{\Delta}_3; \bot.
\end{aligned}
\]
The variable $\mathrm{B}_{2 \delta - 1}$ refers to the rightmost bit of a node's color in forest $F_\delta$ and $\mathrm{B}_{2 \delta}$ refers to the leftmost bit.

In the final phase of the algorithm, each node will compare its color to the colors of its neighbors. We already have variables for the bits of the node's own color, but we also have to define variables for all the bits of its neighbor's colors. Just like in sections \ref{Forest decomposition} and \ref{Cole-Vishkin}, the bits will rotate as we compare each node's color with those of its neighbors. For this purpose, we need $2 \Delta$ timer variables. After comparing the bits, the nodes change color and compare them again. The timer needs to be reset each time and we use a flag $\mathrm{CR}'$ to reset them (which we will define later):
\[
\mathrm{T}^1_1 \coloneq_{\mathrm{CR}', \mathrm{T}^1_1} \top; \top; \bot,
\quad
\mathrm{T}^1_2 \coloneq_{\mathrm{CR}', \mathrm{T}^1_1} \bot; \top; \bot,
\quad
\dots,
\quad
\mathrm{T}^1_{2 \Delta} \coloneq_{\mathrm{CR}', \mathrm{T}^1_{2 \Delta - 1}} \bot; \top; \bot.
\]
We add another $\Delta + 1$ timer variables on top of this for a phase where we determine the lowest available color that a node can switch to, giving us a total of $3 \Delta + 1$ timer variables:
\[
\mathrm{T}^2_1 \coloneq_{\mathrm{CR}', \mathrm{T}^1_{2 \Delta}} \bot; \top; \bot,
\quad
\mathrm{T}^2_2 \coloneq_{\mathrm{CR}', \mathrm{T}^2_1} \bot; \top; \bot,
\quad
\dots,
\quad
\mathrm{T}^2_{\Delta + 1} \coloneq_{\mathrm{CR}', \mathrm{T}^2_{\Delta}} \bot; \top; \bot.
\]

On top of this two-part timer, we need a second timer to count the loops of the first timer. The maximum color of a node in a graph is reduced by at least one in every loop, so we need $t = 3^{\Delta} - \Delta$ variables; after $t - 1$ loops we will have reduced the color of each node to at most $\Delta + 1$. We define them as follows:
\[
\mathrm{LP}_1 \coloneq_{\mathrm{LP}_1, \mathrm{CR'}} \top; \top; \bot,
\quad
\mathrm{LP}_2 \coloneq_{\mathrm{LP}_2, \mathrm{CC}} \top; \mathrm{LP}_1; \bot,
\quad
\dots,
\quad
\mathrm{LP}_t \coloneq_{\mathrm{LP}_t, \mathrm{CC}} \top; \mathrm{LP}_{t - 1}; \bot.
\]
The lower index refers to the current loop of the counter. Once $\mathrm{END4} \coloneqq \mathrm{LP}_t$ becomes true, the program can be stopped. We will use this flag to stop all variables, but we omit it from the definitions to save space.

We define two flags: $\mathrm{CC}$ for the round when nodes change color and $\mathrm{CR}'$ for the global communication rounds right after.
\[
    \mathrm{CC} \coloneq_{\mathrm{CR}' \lor \mathrm{CC}, \mathrm{T}^2_{\Delta + 1}} \bot; \top; \bot,
    \quad
    \mathrm{CR}' \coloneq_{\mathrm{CC}, \mathrm{CR}' \lor \mathrm{T}^1_1, \mathrm{END3}} \top; \bot; \top; \bot.
\]
The condition $\mathrm{END3}$ triggers the first global communication round; this condition will never be revisited. After this, the timer variables begin ticking. Once they have all turned on, the variable $\mathrm{CC}$ becomes true for one round when the nodes change color. Following this, the variable $\mathrm{CR}'$ becomes true for one round as well, marking the global communication round. After this, the counter resets and we start over.

We revise the definitions of all bit variables $\mathrm{B}_i$ to rotate while the first part of the counter is active:
\[
\begin{aligned}
    &\mathrm{B}_1 \coloneq_{\mathrm{T}^2_1, \mathrm{T}^1_1, \mathrm{END3}} \mathrm{B}_1; \mathrm{B}_{2 \Delta}; \mathrm{C}^1_1 \lor \mathrm{C}^1_3; \bot, \\
    &\mathrm{B}_2 \coloneq_{\mathrm{T}^2_1, \mathrm{T}^1_1, \mathrm{END3}} \mathrm{B}_2; \mathrm{B}_{1}; \mathrm{C}^1_2 \lor \mathrm{C}^1_3; \bot, \\
    &\vdots \\
    &\mathrm{B}_{2 \Delta} \coloneq_{\mathrm{T}^2_1, \mathrm{T}^1_1, \mathrm{END3}} \mathrm{B}_{2 \Delta}; \mathrm{B}_{2 \Delta - 1}; \mathrm{C}^{\Delta}_2 \lor \mathrm{C}^{\Delta}_3; \bot.
\end{aligned}
\]
The variables rotate truth values for exactly $2 \Delta$ iterations, stopping at their original values. They remain the same while we search for the smallest available color during the next $\Delta + 1$ iterations. We will later revise these definitions one more time to show how the nodes change colors.

We also need variables $\mathrm{B'}^{\delta}_i$ ($1 \leq \delta \leq \Delta$) for the bits of each neighbor's colors:
\[
\begin{aligned}
    &\mathrm{B'}^{\delta}_1 \coloneq_{\mathrm{CR'}, \mathrm{T}^2_1, \mathrm{T}^1_1} \Diamond_{\delta} \mathrm{B}_1; \mathrm{B'}^{\delta}_1; \mathrm{B'}^{\delta}_{2 \Delta}; \bot, \\
    &\mathrm{B'}^{\delta}_2 \coloneq_{\mathrm{CR'}, \mathrm{T}^2_1, \mathrm{T}^1_1} \Diamond_{\delta} \mathrm{B}_2; \mathrm{B'}^{\delta}_2; \mathrm{B'}^{\delta}_{1}; \bot, \\
    &\vdots \\
    &\mathrm{B'}^{\delta}_{2 \Delta} \coloneq_{\mathrm{CR'}, \mathrm{T}^2_1, \mathrm{T}^1_1} \Diamond_{\delta} \mathrm{B}_{2 \Delta}; \mathrm{B'}^{\delta}_{2 \Delta}; \mathrm{B'}^{\delta}_{2 \Delta - 1}; \bot.
\end{aligned}
\]
The bits rotate just like before. During global communication rounds, the variables receive the bits of the neighbor's updated color.

Next, we define variables $\mathrm{DIF'}^{\delta}$ ($1 \leq \delta \leq \Delta$) that compare the bits of a node with those of its neighbors. We are comparing the sizes of the colors, which means we must start by comparing the leftmost bits:
\[
    \mathrm{DIF'}^{\delta} \coloneq_{\mathrm{CR'}, \mathrm{T}^1_1} \bot; \neg \left( \mathrm{B}_{2 \Delta} \leftrightarrow \mathrm{B'}^{\delta}_{2 \Delta} \right); \bot.
\]
We defined a similar variable in section \ref{Cole-Vishkin}. Similarly, we define variables $\mathrm{HIGH'}^{\delta}$ and $\mathrm{LOW'}^{\delta}$ ($1 \leq \delta \leq \Delta$) that store the information on which node has a higher color:
\[
\begin{aligned}
    &\mathrm{HIGH'}^{\delta} \coloneq_{\mathrm{CR'}, \mathrm{HIGH'}^{\delta}, \mathrm{LOW'}^{\delta}, \mathrm{DIF'}^{\delta}} \bot; \top; \bot; \mathrm{B}_{1}; \bot, \\
    &\mathrm{LOW'}^{\delta} \coloneq_{\mathrm{CR'}, \mathrm{LOW'}^{\delta}, \mathrm{HIGH'}^{\delta}, \mathrm{DIF'}^{\delta}} \bot; \top; \bot; \neg \mathrm{B}_{1}; \bot.
\end{aligned}
\]
The variable $\mathrm{HIGH'}^{\delta}$ becomes true if a node has a higher color than its $\delta$th neighbor and $\mathrm{LOW'}^{\delta}$ becomes true if it has a lower color. The truth value is calculated from the bit on the opposite end (index $1$) than where they were compared (index $2 \Delta$), because the truth value of those variables rotated to the other side during the single round of iteration in between.

Just like in section \ref{Shift-down}, we define a formula that tells us if a node has a greater color than all of its neighbors:
\[
\mathrm{G} \coloneq_{\mathrm{CC}, \mathrm{T}^2_{\Delta + 1}} \bot; \bigwedge\limits_{\delta = 1}^{\Delta} \mathrm{HIGH'}^{\delta}; \bot.
\]

If a node has a greater color than all of its neighbors, it needs to change its color to the lowest color that isn't shared by any of its neighbors. For this purpose, we define $\Delta + 1$ variables that express that a node should change to that color:
\[
\begin{aligned}
    &\mathrm{L}_1 \coloneq_{\mathrm{T}^2_1} \bigwedge\limits_{\delta = 1}^{\Delta} \neg \left( \bigwedge\limits_{i = 2}^{\log(\Delta + 1)} \neg \mathrm{B'}^{\delta}_i \land \mathrm{B'}^{\delta}_1 \right); \bot, \\
    &\mathrm{L}_2 \coloneq_{\mathrm{T}^2_2} \bigwedge\limits_{\delta = 1}^{\Delta} \neg \left( \bigwedge\limits_{i = 3}^{\log(\Delta + 1)} \neg \mathrm{B}'^{\delta}_i \land \mathrm{B'}^{\delta}_2 \land \neg \mathrm{B'}^{\delta}_1 \right) \land \neg \mathrm{L}_1; \bot, \\
    &\mathrm{L}_{3} \coloneq_{\mathrm{T}^2_3} \bigwedge\limits_{\delta = 1}^{\Delta} \neg \left( \bigwedge\limits_{i = 3}^{\log(\Delta + 1)} \neg \mathrm{B'}^{\delta}_i \land \mathrm{B'}^{\delta}_2 \land \mathrm{B'}^{\delta}_1 \right) \land \bigwedge\limits_{\delta = 1}^{2} \neg \mathrm{L}_{\delta}; \bot, \\
    &\vdots
\end{aligned}
\]
In the same iteration round where the corresponding timer variable $\mathrm{T}^2_i$ turns true, each variable $\mathrm{L}_i$ calculates whether the color in its lower index is the lowest available color. In the right conjunction, they check that all lower colors are already taken. In the left conjunction, they check if this color is available using its bits. After $\Delta + 1$ iterations, all variables have calculated a truth value, such that only one of them has become true.

At last, we revise the definitions for the variables $\mathrm{B}_i$ for when nodes change color (we denote $\ell' = \log(\Delta + 1)$, $f(x) = \lceil \frac{x}{2} \rceil$ and $g(x) = x - 2 f(x) + 2$; the functions $f$ and $g$ are used to fix certain indices):
\[
\begin{aligned}
    &\mathrm{B}_1 \coloneq_{\mathrm{G}, \mathrm{T}^2_1, \mathrm{T}^1_1, \mathrm{END3}} \bigvee\limits_{i = 1}^{2^{\ell' - 1}} \mathrm{L}_{2 i - 1}; \mathrm{B}_1; \mathrm{B}_{2 \Delta}; \mathrm{C}^1_1 \lor \mathrm{C}^1_3; \bot, \\
    &\mathrm{B}_2 \coloneq_{\mathrm{G}, \mathrm{T}^2_1, \mathrm{T}^1_1, \mathrm{END3}} \bigvee\limits_{j = 1}^{2} \bigvee\limits_{i = 1}^{2^{\ell' - 2}} \mathrm{L}_{4 i - j}; \mathrm{B}_2; \mathrm{B}_{1}; \mathrm{C}^1_2 \lor \mathrm{C}^1_3; \bot, \\
    &\vdots \\
    &\mathrm{B}_{\ell'} \coloneq_{\mathrm{G}, \mathrm{T}^2_1, \mathrm{T}^1_1, \mathrm{END3}} \bigvee\limits_{j = 1}^{2^{\ell' - 1}} \mathrm{L}_{2^{\ell'} - j}; \mathrm{B}_{\ell'}; \mathrm{B}_{\ell' - 1}; \mathrm{C}^{f(\ell')}_{g(\ell')} \lor \mathrm{C}^{f(\ell')}_3; \bot, \\
    &\mathrm{B}_{\ell' + 1} \coloneq_{\mathrm{G}, \mathrm{T}^2_1, \mathrm{T}^1_1, \mathrm{END3}} \bot; \mathrm{B}_{\ell' + 1}; \mathrm{B}_{\ell'}; \mathrm{C}^{f(\ell' + 1)}_{g(\ell' + 1)} \lor \mathrm{C}^{f(\ell' + 1)}_3; \bot, \\
    &\vdots \\
    &\mathrm{B}_{2 \Delta} \coloneq_{\mathrm{G}, \mathrm{T}^2_1, \mathrm{T}^1_1, \mathrm{END3}} \bot; \mathrm{B}_{2 \Delta}; \mathrm{B}_{2 \Delta - 1}; \mathrm{C}^{\Delta}_2 \lor \mathrm{C}^{\Delta}_3; \bot.
\end{aligned}
\]
During the round when nodes change color, if a node does not have the highest color, its bits stay the same. If it has the highest color, then it changes to a color between $1$ and $\Delta + 1$. This means that bits $\ell' + 1$ onward become $0$ and the corresponding variables become untrue. The truth values of the remaining $\ell'$ variables are calculated from the color a node changes into:
\begin{itemize}
    \item The colors where the rightmost bit is $1$ are $[1] , [3] , [5] , \dots$.
    \item The colors where the second rightmost bit is $1$ are $[2, 3] , [6, 7] , [10, 11], \dots$.
    \item The colors where the third rightmost bit is $1$ are $[4, 5, 6, 7]$ , $[12, 13, 14, 15]$ ,\\ $[20, 21, 22, 23] , \dots$.
    \item The colors where the $k$th rightmost bit is $1$ are the colors $2^k i - j$, where $i \in \Z_+$ and $j \in [k]$.
\end{itemize}

Finally, we define $\Delta + 1$ appointed predicates, one for each possible final color:
\[
\begin{aligned}
    &\mathrm{CLR}_1 \coloneq_{\mathrm{END4}} \bigwedge\limits_{i = 2}^{\log(\Delta + 1)} \neg \mathrm{B}_i \land \mathrm{B}_1; \bot, \\
    &\mathrm{CLR}_2 \coloneq_{\mathrm{END4}} \bigwedge\limits_{i = 3}^{\log(\Delta + 1)} \neg \mathrm{B}_i \land \mathrm{B}_2 \land \neg \mathrm{B}_1; \bot, \\
    &\mathrm{CLR}_3 \coloneq_{\mathrm{END4}} \bigwedge\limits_{i = 3}^{\log(\Delta + 1)} \neg \mathrm{B}_i \land \mathrm{B}_2 \land \mathrm{B}_1; \bot, \\
    &\vdots
\end{aligned}
\]

With these appointed predicates, Theorem \ref{Cole-Vishkin theorem} is now true:

\bigskip 

\noindent
$\blacktriangleright$ \textbf{Theorem\ \ref{Cole-Vishkin theorem}.}\ \emph{There exists a formula of $\mpmsc$ with the following properties.}
\begin{enumerate}
    \item \emph{It defines a proper $(\Delta + 1)$-coloring.}
    \item \emph{It has $\mathcal{O}(\Delta \log(n)) + \mathcal{O}(3^\Delta)$ heads and its size is $\mathcal{O}(\Delta^2 \log(n)) + \mathcal{O}(3^\Delta)$.}
    \item \emph{The number of iterations needed is $\mathcal{O}(\log(n) \log\log(n) \log^*(n)) + \mathcal{O}(\Delta 3^{\Delta})$.}
    \item \emph{There are exactly $\log^*(n) + 10 + 3^\Delta - \Delta$ global communication rounds.}
\end{enumerate}

\begin{proof}
    1) We have $\Delta + 1$ appointed (attention and print) predicates: $\mathrm{CLR}_1, \dots, \mathrm{CLR}_{\Delta + 1}$. They all correspond to different combinations of the variables $\mathrm{B}_1, \dots, \mathrm{B}_{\ell'}$. This makes the print predicates mutually exclusive, meaning that at most one of them can be true in a given node. At least one of the variables $\mathrm{B}_i$ must be true in each node, and by the time the nodes output, the variables $\mathrm{B}_{\ell' + 1}$ onward have all become untrue in every node, which means that at least one print predicate becomes true in each node. Thus, each node outputs a ($\Delta + 1$)-bit string where exactly one bit is a $1$, and the program defines a $(\Delta + 1)$-coloring.
    
    If two neighbors shared the same output, they would share the same truth values for variables $\mathrm{B}_1, \dots, \mathrm{B}_{\ell'}$. They would also share the same truth values for variables $\mathrm{B}_{\ell' + 1}$ onward, because they have all become untrue in every node. This means that the neighbors would share the same color, which is impossible according to the program. Thus, no two neighbors share the same output, and the coloring is proper.

    2) According to Lemma \ref{Shift-down lemma}, we have $\mathcal{O}(\Delta \log(n))$ heads from before and the size of the program was $\mathcal{O}(\Delta^2 \log(n))$. Let's add the new heads and their sizes.
    \begin{itemize}
        \item We have $\mathcal{O}(\Delta)$ variables $\mathrm{T}^i_j$, $\mathrm{CC}$, $\mathrm{CR'}$, $\mathrm{B'}^\delta_{i}$, $\mathrm{DIF'}^\delta$, $\mathrm{HIGH'}^\delta$, $\mathrm{LOW'}^\delta$, $\mathrm{G}$, $\mathrm{L}_i$, $\mathrm{B}_{i}$ and $\mathrm{CLR}_{i}$ of size $\mathcal{O}(\Delta^2)$.
        \item We have $\mathcal{O}(3^\Delta)$ variables $\mathrm{LP}_i$ of size $\mathcal{O}(1)$.
    \end{itemize}
    Adding the heads together with the previous parts of the program, we get
    \[
        \mathcal{O}(\Delta \log(n)) + \mathcal{O}(\Delta) + \mathcal{O}(3^\Delta) = \mathcal{O}(\Delta \log(n)) + \mathcal{O}(3^\Delta).
    \]
    Multiplying the heads by their sizes we get
    \[
    \begin{aligned}
        &\mathcal{O}(\Delta^2 \log(n)) + \mathcal{O}(\Delta) \mathcal{O}(\Delta^2) + \mathcal{O}(3^\Delta) \mathcal{O}(1) \\
        &= \mathcal{O}(\Delta^2 \log(n)) + \mathcal{O}(\Delta^3) + \mathcal{O}(3^\Delta) \\
        &= \mathcal{O}(\Delta^2 \log(n)) + \mathcal{O}(3^\Delta).
    \end{aligned}
    \]

    3) The previous phases took $\mathcal{O}(\log(n) \log\log(n) \log^*(n))$ iterations. The basic color reduction phase takes almost $3^{\Delta}$ repeats of a multiple of $\Delta$ iterations each. The total number of iterations is
    \[
        \mathcal{O}(\log(n) \log\log(n) \log^*(n)) + \mathcal{O}(\Delta) \mathcal{O}(3^{\Delta}) = \mathcal{O}(\log(n) \log\log(n) \log^*(n)) + \mathcal{O}(\Delta 3^{\Delta}).
    \]

    4) There were $\log^*(n) + 12$ previous iterations in the program that required a diamond. In this phase, a diamond is used each time the variable $\mathrm{CR'}$ is true, which happens once each loop. The number of loops is  $3^\Delta - (\Delta + 1)$, which matches the communication rounds of the algorithm.
    This means that the total number of global communication rounds is $\log^*(n) + 3^\Delta - \Delta + 11$.
\end{proof}

If we assume that $\Delta$ is a constant, then the number of heads and the size of the program is $\mathcal{O}(\log(n))$, the number of iterations required is $\mathcal{O}(\log(n) \log\log(n) \log^*(n))$ and the number of global communication rounds is $\log^*(n) + \mathcal{O}(1)$.

\end{document}